\def\year{2021}\relax
\documentclass[letterpaper]{article} 
\usepackage{aaai21}  
\usepackage{times}  
\usepackage{helvet} 
\usepackage{courier}  
\usepackage[hyphens]{url}  
\usepackage{graphicx} 
\usepackage{natbib}  
\usepackage{caption} 
\usepackage{amsthm}
\usepackage{amsmath}
\usepackage{amssymb}
\usepackage{mathtools}
\usepackage[switch]{lineno}  %
\usepackage{tikz}
\usetikzlibrary{arrows,backgrounds}
\usepackage{pgfplots}
\pgfplotsset{width=10cm,compat=1.9}
\usetikzlibrary{plothandlers}
\usetikzlibrary{plotmarks}

\newtheorem{lemma}{Lemma}
\newtheorem{theorem}{Theorem}

\theoremstyle{definition}
\newtheorem{definition}{Definition}
\newtheorem{example}{Example}
\newtheorem{innercustomthm}{Theorem}
\newenvironment{customthm}[1]
  {\renewcommand\theinnercustomthm{#1}\innercustomthm}
  {\endinnercustomthm}
\newtheorem{innercustomlemma}{Lemma}
\newenvironment{customlemma}[1]
  {\renewcommand\theinnercustomlemma{#1}\innercustomlemma}
  {\endinnercustomlemma}
\newtheorem{innercustomexample}{Example}
\newenvironment{customexample}[1]
  {\renewcommand\theinnercustomexample{#1}\innercustomexample}
  {\endinnercustomexample}  
\newtheorem{innercustomcorollary}{Corollary}
\newenvironment{customcorollary}[1]
  {\renewcommand\theinnercustomcorollary{#1}\innercustomcorollary}
  {\endinnercustomcorollary}
\title{Fair and Efficient Allocations with Limited Demands  }
\author{

    Sushirdeep Narayana, Ian A. Kash \\
}
\affiliations{

    snaray25@uic.edu, iankash@uic.edu, \\
    Department of Computer Science, \\
    University of Illinois at Chicago, \\
    Chicago, Illinois, United States

}

\begin{document}

\maketitle

\begin{abstract}
We study the fair division problem of allocating multiple resources among a set of agents with Leontief preferences that are each required to complete a finite amount of work, which we term ``limited demands.'' We examine the behavior of the classic Dominant Resource Fairness (DRF) mechanism in this setting and show it is fair but only weakly Pareto optimal and inefficient in many natural examples. We propose as an alternative the Least Cost Product (LCP) mechanism, a natural adaptation of Maximum Nash Welfare to this setting. 
We characterize the structure of allocations of the LCP mechanism in this setting, show that it is Pareto efficient, and that it satisfies the relatively weak fairness property of sharing incentives.  While we prove it satisfies the stronger fairness property of (expected) envy freeness in some special cases, we provide a counterexample showing it does not do so in general, a striking contrast to the ``unreasonable fairness'' of Maximum Nash Welfare in other settings.  Simulations suggest, however, that these violations of envy freeness are rare in randomly generated examples.  
\end{abstract}

\noindent
\section{Introduction}

\noindent Large computational tasks, including those found in cloud computing systems, data centers, and high-performance computing clusters, require multiple heterogeneous resources such as CPU, memory, and network bandwidth.  For these systems to be effective, it is important that these and other resources be allocated among various tasks in a fair and efficient manner.

One popular approach to resource allocation in this setting is the Dominant Resource Fairness (DRF) mechanism, which achieves a number of desirable properties including efficiency (Pareto optimality), fairness (sharing incentives and envy-freeness), and strategy-proofness~\cite{Ghodsi2011drf,parkes2015beyond}.  However, these results are obtained under the assumption that the utility of each task is determined by the amount of resources it receives.  We argue that a more natural model for many applications is that they have limited demands: each task has a finite amount of work to do and would like to complete as soon as possible. Such a model can be applied to the frameworks of job schedulers of modern computing clusters  such as those used in \citet{Hindman2011Mesos}, \citet{vavilapalli2013apacheHadoopYARN}, and \citet{Grandl2016Carbyne}.

The addition of limited demands does not change the fairness and incentive properties, but has a substantial effect on efficiency.  Consider the following simple example.  There are two agents and one (divisible) resource, and each agent needs a single unit of the resource to complete its task.  The DRF allocation evenly shares the resource between the two agents, so both complete at time 2.  Instead, a shortest job first approach will have one agent complete at time 1 and the other complete at time 2.  This makes one agent strictly better off and the other no worse off showing that the result of DRF is no longer Pareto optimal in this setting.  Of course, this allocation is not envy-free, as the agent scheduled second envies the agent scheduled first.  Nevertheless, if we randomize which agent is scheduled first, it would be envy-free in expectation.

While shortest job first is envy-free in expectation and Pareto efficient, with multiple resources it too can pass up opportunities to make desirable trade-offs.  Consider two tasks of similar length: one memory-limited and one CPU-limited.  Shortest job first will pass up the opportunity to run both at the same time, which would slightly slow down the first completion time but substantially speed up the second.  By most metrics of interest (e.g. mean completion time or makespan) this is an improvement.

Thus, a natural goal is an algorithm that can provide a trade off between sharing and shortest job first to achieve efficiency as well as fairness.  
Prior systems work (\citet{Grandl2016Carbyne}) has developed heuristic solutions to this problem; our contribution is a principled fair division approach.
We propose the Least Cost Product (LCP) mechanism, which minimizes the product of completion times (costs) in the same way Maximum Nash Welfare (MNW) mechanisms maximize the product of utilities.\footnote{In fact, if utility is defined to be the reciprocal of completion time, our mechanism chooses exactly the MNW  solution.} We provide a characterization of the structure of solutions it finds and show that the LCP mechanism satisfies sharing incentives and is Pareto optimal but not strategy-proof. We show that it satisfies envy-freeness in expectation for two special cases: when there is only a single resource or when there are only two agents. When more agents or resources are involved, we present an example of an LCP allocation with envy. This is a striking contrast to the fairness properties satisfied by the Maximum Nash Welfare allocations in other settings. All adaptations of Maximum Nash Welfare we are aware of in divisible settings, such as Competitive Equilibrium with Equal Incomes (CEEI) for dividing divisible goods with linear additive utilities  \cite{VARIAN197463, moulin2004fair} and Leontief utilities \cite{Ghodsi2011drf, goel2019markets}; the competitive rule for dividing bads with linear additive disutilities \cite{bogomolnaia2019dividing}; and competitive allocations for dividing mixed manna with concave, continuous and 1-homothetic (dis)utilities \cite{bogomolnaia2017competitive} satisfy strong envy-freeness guarantees. We believe our setting of ``limited demands'' is the first natural divisible setting where an adaptation of Maximum Nash Welfare has resulted in allocations not being envy-free. However, we show in simulations on randomly generated instances that the LCP allocations that violate envy-freeness are rare.

\section{Related Work} \label{related-work}

\citeauthor{Ghodsi2011drf} \shortcite{Ghodsi2011drf} introduced the Dominant Resource Fairness (DRF) mechanism for allocation of resources when the agents have Leontief preferences. \citeauthor{parkes2015beyond} \shortcite{parkes2015beyond} later extended the DRF mechanism in several directions. They generalized the DRF mechanism to more expressive settings and established stronger fairness properties like Group Strategy-Proofness. \citeauthor{kash2014no} \shortcite{kash2014no} extended the DRF mechanism for dynamic settings, where agents arrive with their demands dynamically in an online manner. These works do not consider any notion of a limited amount of resource required by an agent to complete its work.

\citeauthor{VARIAN197463} \shortcite{VARIAN197463} explored the CEEI mechanism which is a market-based  interpretation to allocate goods and showed that the allocations are envy-free and Pareto efficient. \citeauthor{arrow2000handbook} \shortcite{arrow2000handbook} found that the CEEI allocations coincide with allocations where the objective is to maximize the Nash welfare. \citeauthor{Caragiannis2016MaxNashWelfare} \shortcite{Caragiannis2016MaxNashWelfare} showed that for allocating indivisible items the maximum Nash welfare (MNW) solution achieves both EF1 (envy-freeness up-to one good) and Pareto optimality simultaneously, motivating them to call it \textit{unreasonably fair}. For a survey, see~\citet{moulin2019fair}. \citet{Conitzer2017PublicFairDivision} worked on the fair division problem for public decision making and showed that the MNW solution approximates or satisfies various notions of proportionality. \citeauthor{conitzer2019group} \shortcite{conitzer2019group} introduced the notion of \textit{group fairness} which generalizes envy freeness to groups of agents. They showed that locally optimal Nash welfare allocations satisfy "up-to one good" style relaxations of \textit{group fairness}. These works illustrate the promising nature of maximizing Nash welfare. However, none of these works adapt the MNW solution to our setting and none of them observe the failures of fairness we do.

More closely related to our work,
\citeauthor{friedman2003fairness} \shortcite{friedman2003fairness} analyzed the problem of single resource allocation when agents with jobs that required different service times arrived in an online manner. They proposed the Fair Sojourn Protocol which is more efficient and fair compared to the Shortest Remaining Processing Time and Processor Sharing protocols (corresponding to Shortest Job First and DRF in our setting respectively). We explore a generalized, but offline version of their problem where agents have demands on multiple heterogeneous resources. 

\citeauthor{Grandl2016Carbyne} \shortcite{Grandl2016Carbyne} proposed an \textit{altruistic} two-level scheduler called Carbyne that schedules jobs demanding multiple resources where each job takes the form of a Directed Acyclic Graph (DAG) with vertices representing tasks, and edges denoting dependencies. They show that experimentally Carbyne performs better than the DRF and the Shortest Job First (SJF) when evaluating inter-job fairness, performance and cluster efficiency. This work inspired our examination of limited demands.

Finally, \citeauthor{bogomolnaia2017competitive} \shortcite{bogomolnaia2017competitive} generalized the competitive division mechanism for the fair division of mixed manna where items are goods to some agents, but bads or satiated to others. They prove that when the items consist of only bad utility profiles, the competitive profiles are the critical points of the product of disutilities on the efficiency frontier. Our characterization of the LCP mechanism, which is also an allocation of bads, has a similar flavor.

\section{Preliminaries}   \label{prelim}

Let $\mathcal{N} = \{1,2,3,...,n\}$ denote the set of agents and $\mathcal{R} = \{1, 2, ..., m\}$ denote the set of heterogeneous resources available in the system. The resources are assumed to be divisible. Each resource in the set $\mathcal{R}$ has a unit capacity. An agent in the system demands various resources be allocated to it in fixed proportions, known as Leontief preferences. Let $\mathbf{D}$ represent the demands of all the agents with each element $d_{ir}$ denoting the fraction of resource \textit{r} required by agent \textit{i}. We assume these are normalized so that the demand vector for agent \textit{i} denoted by $\mathbf{d_{i}} =  \langle d_{i1}, d_{i2}, ..., d_{im} \rangle$ has $d_{ir} = 1$ for some $r$. The number of tasks that are required to be completed by an agent \textit{i} is $k_{i}$ and is represented by an $n \times n$ diagonal matrix $\mathbf{K}$ with the \textit{i}-th diagonal element equal to $k_{i}$.\footnote{The normalization of $d_i$ means that $k_i$ need not be an integer.} The amount of resources that needs to be consumed by an agent \textit{i} in order to complete its job is given by $\mathbf{w_{i}} = k_{i}\mathbf{d_{i}}$. In matrix notation, this becomes $\mathbf{W} = \mathbf{K} \mathbf{D}$. 

\begin{definition}
An instantaneous allocation $\mathbf{A}
\subseteq \mathbb{R}^{n \times m}$ allocates a portion $A_{ir}$ of resource type \textit{r} to agent \textit{i} subject to the feasibility conditions, $ \displaystyle \sum_{i \in \mathcal{N}} A_{ir} \leq 1, \; \forall r \in \mathcal{R}$; and $A_{i} = \lambda_{i} \mathbf{d_{i}}, \; \forall i \in \mathcal{N} $ and for some $\lambda_{i} \in \mathbb{R}_{\geq 0}$.
\end{definition}

\begin{definition}
A resource allocation mechanism is a function \textit{f} that takes as inputs the amount of resources $\mathbf{W}$ that are required by the agents and maps them to allocations as outputs, that is, $f : \mathbb{R}^{n \times m} \rightarrow (\mathbb{R}^{n \times m})^{\mathbb{R}_{\geq 0}} $ such that $f_{i}(\mathbf{W}) = \mathbf{A}_{i}(\cdot)$, where $\mathbf{A}_{i}(t)$ denotes the allocation of agent $i$ at time $t$, which we refer to as its {\em instantaneous allocation}.
\end{definition}

This definition assumes that resources are fully divisible, which is not necessarily true in real systems.  In practice this can be handled by keeping allocations as close as possible to the divisible ideal over time~\cite{Hindman2011Mesos}.

The cost of an agent characterizes the loss suffered by that agent and is equal to the completion time of the agent for its resource allocation. The cost $c_{i}(f_{i}(\mathbf{W})) = t_{i}$ for an agent \textit{i} under the allocation mechanism \textit{f} can be calculated as the value of $t_i$ that solves $ \mathbf{w_{i}} = \displaystyle \int_{0}^{t_{i}}\mathbf{A}_{i}(t)dt$. 

In order to evaluate and compare the allocations produced by various resource allocation mechanisms, we define four standard fairness and efficiency properties for our setting:
\begin{itemize}
    \item Pareto Optimality (PO): An allocation $\mathbf{A}(\cdot)$ is said to be Pareto optimal if there is no alternative allocation $\mathbf{A}^{'}(\cdot)$  which can make at least one agent strictly better off without making any other agent worse off. Formally, \\ $\forall \mathbf{A}^{'}(\cdot),\quad 
    (\exists i \in \mathcal{N}, \quad c_{i}(\mathbf{A}^{'}_{i}(\cdot)) < c_{i}(\mathbf{A}_{i}(\cdot))) \\ \implies  (\exists j \in \mathcal{N}, \quad c_{j}(\mathbf{A}^{'}_{j}(\cdot)) > c_{j}(\mathbf{A}_{j}(\cdot)))$.
    \\An allocation $\mathbf{A}(\cdot)$  is weakly Pareto optimal if there is no alternative allocation $\mathbf{A}^{'}(\cdot)$ which would strictly benefit all the agents. That is,\\ $\nexists \mathbf{A}^{'}(\cdot), \;$ where $\forall i \in \mathcal{N}, \quad c_{i}(\mathbf{A}^{'}_{i}(\cdot)) \; < \; c_{i}(\mathbf{A}_{i}(\cdot))$
    
    \item Sharing Incentives (SI): By abuse of notation, we will refer to $\mathbf{A}_{\textnormal{SI}}(\cdot) = \mathbf{A}_{\textnormal{SI}} = \; \langle \frac{1}{n}, \frac{1}{n}, ..., \frac{1}{n} \rangle$ as the instantaneous allocation that splits all the \textit{m} resources equally among all the agents. An allocation mechanism satisfies Sharing Incentives if, \quad 
    $c_{i}(\mathbf{A}_{i}(\cdot)) \; \leq \; c_{i}(\mathbf{A}_{\textnormal{SI}}) = n \cdot k_{i}, \quad \forall i \in \mathcal{N}$. In other words, if the cost for each agent's allocation under the allocation mechanism is at-most the cost encountered when there is an equal split of the resources among agents, the allocation mechanism satisfies SI.
    
    \item  Envy-Freeness (EF): An allocation mechanism is envy-free (EF) if 
    \quad $\forall i,j \in \mathcal{N}, \quad c_{i}(\mathbf{A}_{i}(\cdot)) \; \leq \; c_{i}(\mathbf{A}_{j}(\cdot))$. That is, each agent would never strictly prefer the resource allocation received by another agent.
    
    A randomized allocation mechanism is envy-free in expectation if,
    $\forall i,j \in \mathcal{N}, \quad \mathbb{E}[c_{i}(\mathbf{A}_{i}(\cdot))] \; \leq \; \mathbb{E}[c_{i}(\mathbf{A}_{j}(\cdot))],$
    where the expectation is taken over the randomness of the mechanism.  In particular, when there are multiple agents with the same amount of work, the allocation mechanism may have multiple optimal solutions.  In these cases, as in the example in the introduction, each individual choice has envy but a uniform random selection from them will be envy-free in expectation.
    
    \item Strategy-Proofness (SP): An allocation mechanism is strategy-proof (SP)  if an agent can never benefit from reporting a false demand, regardless of the demands of the other agents. Given the demands $\mathbf{D}$ and the number of tasks $\mathbf{K}$, let $f(\mathbf{KD})$ be the resulting allocation due to the mechanism \textit{f}. Suppose an agent $i \in \mathcal{N}$ changes to an untruthful demand $\mathbf{d}^{\;'}_{i}$ while the demands of the other agents and the number of tasks $\mathbf{K}$ stay the same. The new demand matrix is denoted by $\mathbf{D}^{'}$. The allocation mechanism \textit{f} is said to be SP if,
    \\$\forall i \in \mathcal{N}, \; \forall d_i, \; c_{i}(f_{i}(\mathbf{KD})) \leq c_{i}(f_{i}(\mathbf{KD}^{'}))$, \\[2 pt]where $\mathbf{D} = \begin{bmatrix} \mathbf{d}_{1} \\ ... \\ \mathbf{d}_{i} \\ ... \\ \mathbf{d}_{n}\end{bmatrix}$, and $\mathbf{D}^{'} = \begin{bmatrix} \mathbf{d}_{1} \\ .. \\ \mathbf{d}^{'}_{i} \\ .. \\ \mathbf{d}_{n}\end{bmatrix}$. 
    
\end{itemize}

\begin{definition}
An allocation is defined to be a \textit{fixed} allocation
when the resources allocated to each agent are fixed over the intervals determined by the different completion times of the agents. The time intervals are given by a sorted order of completion times of the agents (where agents are numbered in increasing order of completion time) $[t_{a}, \; t_{b})$, where $t_{a} \in \{0,t_{1},t_{2},...,t_{n}\}$ and $t_{b} = \underset{i \in \mathcal{N}}{\min} \{t_{i}: t_{i} > t_{a}\}$.
\end{definition}

Proofs omitted from this and subsequent sections can be found in the Appendix. 

\begin{lemma} \label{alloc-behavior}
Given any resource allocation $\mathbf{A}(\cdot)$, there exists an equivalent fixed allocation $\mathbf{A}^{\textnormal{fixed}}(\cdot)$ where all the agents have the same costs/utilities.
\end{lemma}

Hence, without loss of generality, we will consider only fixed allocations for the remainder of the paper. 

\section{Dominant Resource Fairness with Work} \label{drf-w}
The DRF-W mechanism is an adaptation of the DRF mechanism~\cite{Ghodsi2011drf,parkes2015beyond} to our setting.
The DRF mechanism generalizes max-min fairness to settings with multiple goods and Leontief utilities.
It terms the resource satisfying, $r^{*} = \underset{r}{\arg \max} \; d_{ir}$ the {\em dominant resource} of an agent \textit{i}. The share $A_{ir^{*}}$ of dominant resource allocated to agent \textit{i} is its {\em dominant share}. The DRF mechanism chooses the allocation that equalizes the dominant shares of all the agents and is efficiently computable even at large scale~\cite{kash2018dc}.

The mechanism $f^{\textnormal{DRF-W}}$ initially allocates resources to all the agents using the same resource allocation policy used for DRF until one of the agents completes its job. After each agent finishes its work, ${f}^{\textnormal{DRF-W}}$ mechanism reruns DRF on the remaining agents. Thus, at time $t$ the DRF-W allocation assigns each agent $\lambda(t)$ of the agent's dominant resource and corresponding amounts of the other resources where $\lambda(t)$ is the solution to the linear program,
\begin{linenomath}
\begin{align*} \label{eq:1}
\begin{split}
\max \lambda(t) & 
\\ \textnormal{subject to} \underset{i \in \mathcal{N}(t)}{\displaystyle \sum} & \lambda(t) \cdot d_{ir} \leq 1, \quad \forall r \in \mathcal{R}.    
\end{split}
\end{align*}
\end{linenomath}
\noindent Here, $\mathcal{N}(t)$ denotes the set of agents that have not yet completed at time \textit{t}.

We illustrate the DRF-W mechanism with an example. 
\begin{example} \label{drf-example}
Consider a set of two agents where one agent requires 1 GB RAM (resource $r_1$) and 0.5 Mbs network bandwidth (resource $r_2$) and the other agent requires 0.25 GB RAM and 1 Mbs network bandwidth. Both agents have one unit of task to complete. In our notation this becomes 

\noindent $\mathbf{D} = \begin{bmatrix} 1 \quad \dfrac{1}{2} \\\dfrac{1}{4} \quad 1 \end{bmatrix}, \quad$
$\mathbf{K} = \begin{bmatrix} 1 \quad 0 \\0 \quad 1 \end{bmatrix}, \quad \textnormal{and}$
$\mathbf{W} = \begin{bmatrix} 1 \quad \dfrac{1}{2} \\\dfrac{1}{4} \quad 1 \end{bmatrix}$.\\[4 pt]
The DRF-W allocation at the first interval is obtained by solving the following linear program,
\begin{linenomath}
\begin{align*} 
\begin{split}
\max \quad & \lambda  
\\ \textnormal{subject to} \quad  & \lambda \cdot 1 + \lambda \cdot \dfrac{1}{4} \leq 1, \\ &  \lambda \cdot \dfrac{1}{2} + \lambda \cdot 1 \leq 1    
\end{split}
\end{align*}
\end{linenomath}

\noindent where $\lambda$ is the dominant share of both the agents in the first time period. Solving the above linear program gives $\lambda = \displaystyle  \dfrac{2}{3}$, and the DRF-W allocations as
\begin{linenomath}
\begin{align*} 
\begin{split}
\mathbf{A}^{\textnormal{DRF-W}}_{1}(t) & = \langle \dfrac{2}{3}, \; \dfrac{1}{3}\rangle, \quad \forall t \in [0, \dfrac{3}{2}] \\[1 pt]
\mathbf{A}^{\textnormal{DRF-W}}_{2}(t) & = \langle \dfrac{1}{6}, \; \dfrac{2}{3} \rangle, \quad \forall t \in [0, \dfrac{3}{2}]  
\end{split}
\end{align*}
\end{linenomath}

\noindent Both the agents complete at time $\displaystyle \dfrac{3}{2}$. Figure 1 illustrates the DRF-W allocation obtained in Example \ref{drf-example}.
\end{example}


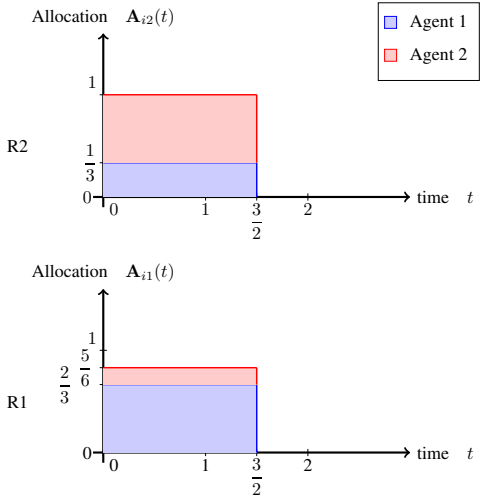
\begin{figure}[!ht]
 \scalebox{0.68}{\begin{tikzpicture}
\draw  node[anchor = west] at (-2, 1) {$\text{R1}$} ;
\draw[very thick, ->] (-0.2, 0) -- (6, 0) node[right] {$\text{time} \quad  t$} ;
\draw[very thick, ->] (0, -0.2) -- (0, 3.2) node[above] {$\text{Allocation} \quad \mathbf{A}_{i1}(t)$} ;
\draw node[anchor = north west] at (0, 0) {$\text{0}$}; 
\draw node[anchor =  east] at (-0.1, 0) {$\text{0}$};

\pgfplothandlermark{\pgfuseplotmark{|}}
\pgfplotstreamstart
\pgfplotstreampoint{\pgfpoint{0cm}{0cm}}
\pgfplotstreampoint{\pgfpoint{2cm}{0cm}}
\pgfplotstreampoint{\pgfpoint{3cm}{0cm}}
\pgfplotstreampoint{\pgfpoint{4cm}{0cm}}
\pgfplotstreamend

\pgfplothandlermark{\pgfuseplotmark{-}}
\pgfplotstreamstart
\pgfplotstreampoint{\pgfpoint{0cm}{0cm}}
\pgfplotstreampoint{\pgfpoint{0cm}{1.33cm}}
\pgfplotstreampoint{\pgfpoint{0cm}{1.66cm}}
\pgfplotstreampoint{\pgfpoint{0cm}{2cm}}
\pgfplotstreamend

\fill[fill= blue!20] (0, 0) rectangle (3, 1.33) ;
\draw[blue, thick] (0, 1.33) -- (3,1.33);
\draw[blue, thick] (3,1.33) -- (3,0);
\draw node[anchor = north] at (2, 0) {$\text{1}$}; 
\draw node[anchor = north] at (3, 0) {$\displaystyle \frac{3}{2}$};
\draw node[anchor = east] at (-0.5, 1.33) {$\displaystyle \dfrac{2}{3}$} ;

\fill[fill = red!20] (0, 1.33) rectangle (3, 1.66) ;
\draw[red, thick] (0, 1.66) -- (3, 1.66) ;
\draw[red, thick] (3, 1.66) -- (3, 1.33); 
\draw node[anchor = north] at (4, 0) {$\text{2}$};
\draw node[anchor = east] at (-0.1, 1.66) {$\displaystyle \frac{5}{6}$} ;
\draw node[anchor = south east] at (0, 2) {$1$} ;

\draw  node[anchor = west] at (-2, 6) {$\text{R2}$} ;
\draw[very thick, ->] (-0.2, 5) -- (6, 5) node[right] {$\text{time} \quad  t$} ;
\draw[very thick, ->] (0, 4.8) -- (0, 8.2) node[above] {$\text{Allocation} \quad \mathbf{A}_{i2}(t)$} ;
\draw node[anchor = north west] at (0,5) {$0$};
\draw node[anchor = east] at (-0.1,5) {$0$} ;

\pgfplothandlermark{\pgfuseplotmark{|}}
\pgfplotstreamstart
\pgfplotstreampoint{\pgfpoint{0cm}{5cm}}
\pgfplotstreampoint{\pgfpoint{2cm}{5cm}}
\pgfplotstreampoint{\pgfpoint{3cm}{5cm}}
\pgfplotstreampoint{\pgfpoint{4cm}{5cm}}
\pgfplotstreamend

\pgfplothandlermark{\pgfuseplotmark{-}}
\pgfplotstreamstart
\pgfplotstreampoint{\pgfpoint{0cm}{5cm}}
\pgfplotstreampoint{\pgfpoint{0cm}{5.67cm}}
\pgfplotstreampoint{\pgfpoint{0cm}{7cm}}
\pgfplotstreamend

\fill[fill= blue!20] (0, 5) rectangle (3, 5.67) ;
\draw[blue, thick] (0, 5.67) -- (3,5.67);
\draw[blue, thick] (3,5.67) -- (3,5);
\draw node[anchor = north] at (2,5) {$1$};
\draw node[anchor = east] at (0,5.67) {$\displaystyle \frac{1}{3}$};
\draw node[anchor = north] at (3,5) {$\displaystyle \frac{3}{2}$};
\draw node[anchor = north] at (4,5) {$2$};

\fill[fill = red!20] (0, 5.67) rectangle (3, 7) ;
\draw[red, thick] (0, 7) -- (3, 7) ;
\draw[red, thick] (3, 7) -- (3, 5.67); 

\draw node[anchor = south east] at (0, 7) {$1$};

\matrix[draw, below left] at (current bounding box.north east)
{\node at ++(-2,0.5) [rectangle, draw = blue, fill= blue!20] {};   \node at ++(-0.5,0.65) [] {Agent 1};\\
\node at ++(-2, 0.5) [rectangle, draw = red, fill = red!20]{};
\node at ++(-0.5, 0.65) [] {Agent 2};\\
} ;

\end{tikzpicture}}
\caption{DRF-W Allocation described in Example \ref{drf-example}}
\end{figure}


As discussed in the introduction, DRF-W is not Pareto optimal.  When compared to the LCP solution given in Example 2 in the next section, this example provides another illustration of this phenomenon.
However, DRF-W is weakly Pareto optimal, in that it is not possible for an alternate allocation to make all agents strictly better off.  The following theorem shows that it satisfies this as well as several other desiderata.

\begin{theorem} \label{drf-w_properties}
DRF-W satisfies weak-PO, SI, EF, and SP.
\end{theorem}

\section{Least Cost Product Mechanism} \label{LCP}

\noindent The Least Cost Product (LCP) mechanism chooses resource allocations such that the product of the costs for the allocations given to agents is the minimum. That is,
\begin{linenomath}
\begin{equation*}
    {f}^{\textnormal{LCP}}(\mathbf{W}) \in \underset{\mathbf{A}(\cdot)}{\arg\min} \; \underset{i \in \mathcal{N}}{\displaystyle \prod}  c_{i}(\mathbf{A}_{i}(\cdot)).
\end{equation*}
\end{linenomath}

\begin{example} \label{lcp-example}
We illustrate the LCP mechanism with the same example used in the previous section. Consider the same instances of $\mathbf{D}, \; \mathbf{K}, \; \textnormal{and} \; \mathbf{W}\;$ used in Example 1 described in the previous section. The completion time of the first agent is given by $t_{1} = \displaystyle \dfrac{k_{1}}{\lambda_{1,1}}$, where $\lambda_{1,1}$ is the dominant share of the first agent. The completion time of the second agent is expressed in terms of the work remaining for the second agent to complete after the completion of the first agent. In the upcoming subsection, we give a more detailed description of expressing the completion time of an agent in Equation (1). The cost of the second agent is $t_{2} = \displaystyle\dfrac{k_{2}}{\lambda_{2,2}} - \dfrac{k_{1}}{\lambda_{1,1}} \cdot \dfrac{\lambda_{2,1}}{\lambda_{2,2}} + \displaystyle \dfrac{k_{1}}{\lambda_{1,1}}$, where $\lambda_{2,1}$ and $\lambda_{2,2}$ are the dominant shares of the second agent at the first and second time intervals respectively.  The LCP allocation is obtained by solving the following optimization problem,
\begin{linenomath}
\begin{align*} 
\begin{split}
\min \quad & \dfrac{1}{\lambda_{1,1}} \cdot \left( 1 - \displaystyle \dfrac{\lambda_{2,1}}{\lambda_{1,1}} + \displaystyle \dfrac{1}{\lambda_{1,1}}\right)  
\\ \textnormal{subject to} \quad  & \lambda_{1,1} \cdot 1 + \lambda_{2,1} \cdot \dfrac{1}{4} \leq 1, \\ &  \lambda_{1,1} \cdot \dfrac{1}{2} + \lambda_{2,1} \cdot 1 \leq 1  
\end{split}
\end{align*}
\end{linenomath}

\noindent where we have substituted $\lambda_{2,2} = 1$ since there are only two agents. Solving the above optimization problem gives us $\lambda_{1,1} = \displaystyle \dfrac{6}{7}, \; \lambda_{2,1} = \displaystyle \dfrac{4}{7}$. The LCP allocations are as follows,
\begin{linenomath}
\begin{align*} 
\begin{split}
\mathbf{A}^{\textnormal{LCP}}_{1}(t) & = \langle \dfrac{6}{7}, \; \dfrac{3}{7} \rangle, \quad \forall t \in [0, \dfrac{7}{6}) \\[1.5 pt]
\mathbf{A}^{\textnormal{LCP}}_{2}(t) &  = \begin{cases}
\langle \dfrac{1}{7}, \; \dfrac{4}{7} \rangle, \quad \forall t \in [0, \dfrac{7}{6}) \\[6 pt]
\langle \dfrac{1}{4}, \; 1 \rangle, \quad \forall t \in [\dfrac{7}{6}, \dfrac{3}{2}] 
\end{cases}
\end{split}
\end{align*}
\end{linenomath}

Figure 2 illustrates the LCP allocation obtained in Example \ref{lcp-example}.
In comparison to the DRF-W allocation in Example 1, this example illustrates the opportunity for improving the welfare of one agent at no cost to the other.
\end{example}


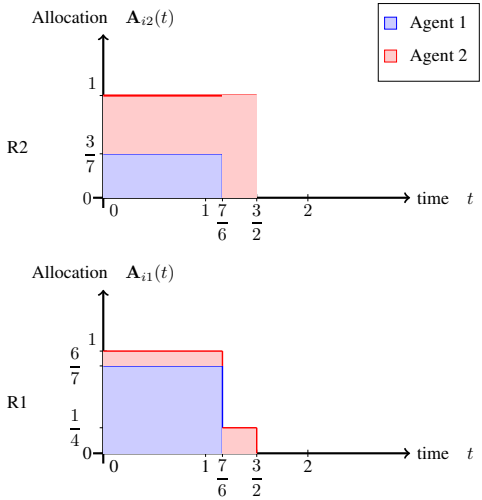
\begin{figure}[!ht]
 \scalebox{0.68}{\begin{tikzpicture}
\draw  node[anchor = west] at (-2, 1) {$\text{R1}$} ;
\draw[very thick, ->] (-0.2, 0) -- (6, 0) node[right] {$\text{time} \quad  t$} ;
\draw[very thick, ->] (0, -0.2) -- (0, 3.2) node[above] {$\text{Allocation} \quad \mathbf{A}_{i1}(t)$} ;
\draw node[anchor = north west] at (0, 0) {$\text{0}$}; 
\draw node[anchor =  east] at (-0.1, 0) {$\text{0}$};

\pgfplothandlermark{\pgfuseplotmark{|}}
\pgfplotstreamstart
\pgfplotstreampoint{\pgfpoint{0cm}{0cm}}
\pgfplotstreampoint{\pgfpoint{2cm}{0cm}}
\pgfplotstreampoint{\pgfpoint{2.33cm}{0cm}}
\pgfplotstreampoint{\pgfpoint{3cm}{0cm}}
\pgfplotstreampoint{\pgfpoint{4cm}{0cm}}
\pgfplotstreamend

\pgfplothandlermark{\pgfuseplotmark{-}}
\pgfplotstreamstart
\pgfplotstreampoint{\pgfpoint{0cm}{0cm}}
\pgfplotstreampoint{\pgfpoint{0cm}{0.5cm}}
\pgfplotstreampoint{\pgfpoint{0cm}{1.71cm}}
\pgfplotstreampoint{\pgfpoint{0cm}{2cm}}
\pgfplotstreamend

\fill[fill= blue!20] (0, 0) rectangle (2.33, 1.71) ;
\draw[blue, thick] (0, 1.71) -- (2.33, 1.71);
\draw[blue, thick] (2.33,1.71) -- (2.33,0);
\draw node[anchor = north] at (2, 0) {$\text{1}$}; 
\draw node[anchor = north] at (2.33, 0) {$\displaystyle \frac{7}{6}$};
\draw node[anchor = east] at (-0.3, 1.71) {$\displaystyle \dfrac{6}{7}$} ;

\fill[fill = red!20] (0, 1.71) rectangle (2.33, 2) ;
\draw[red, thick] (0, 2) -- (2.33, 2) ;
\draw[red, thick] (2.33, 2) -- (2.33, 1.71) ;
\draw node[anchor = south east] at (0, 2) {$1$} ;

\fill[fill = red!20] (2.33, 0) rectangle (3, 0.5) ;
\draw[red, thick] (2.33, 0.5) -- (3, 0.5) ;
\draw[red, thick] (3, 0.5) -- (3, 0) ;
\draw node[anchor = north] at (4, 0) {$\text{2}$};
\draw node[anchor = east] at (-0.3, 0.5) {$\displaystyle \dfrac{1}{4}$} ;
\draw node[anchor = north] at (3, 0) {$\displaystyle \frac{3}{2}$};

\draw  node[anchor = west] at (-2, 6) {$\text{R2}$} ;
\draw[very thick, ->] (-0.2, 5) -- (6, 5) node[right] {$\text{time} \quad  t$} ;
\draw[very thick, ->] (0, 4.8) -- (0, 8.2) node[above] {$\text{Allocation} \quad \mathbf{A}_{i2}(t)$} ;
\draw node[anchor = north west] at (0,5) {$0$};
\draw node[anchor = east] at (-0.1,5) {$0$} ;

\pgfplothandlermark{\pgfuseplotmark{|}}
\pgfplotstreamstart
\pgfplotstreampoint{\pgfpoint{0cm}{5cm}}
\pgfplotstreampoint{\pgfpoint{2cm}{5cm}}
\pgfplotstreampoint{\pgfpoint{2.33cm}{5cm}}
\pgfplotstreampoint{\pgfpoint{3cm}{5cm}}
\pgfplotstreampoint{\pgfpoint{4cm}{5cm}}
\pgfplotstreamend

\pgfplothandlermark{\pgfuseplotmark{-}}
\pgfplotstreamstart
\pgfplotstreampoint{\pgfpoint{0cm}{5cm}}
\pgfplotstreampoint{\pgfpoint{0cm}{5.86cm}}
\pgfplotstreampoint{\pgfpoint{0cm}{7cm}}
\pgfplotstreamend

\fill[fill= blue!20] (0, 5) rectangle (2.33, 5.86) ;
\draw[blue, thick] (0, 5.86) -- (2.33,5.86);
\draw[blue, thick] (2.33,5.86) -- (2.33,5);
\draw node[anchor = north] at (2,5) {$1$};
\draw node[anchor = east] at (0,5.86) {$\displaystyle \frac{3}{7}$};
\draw node[anchor = north] at (2.33,5) {$\displaystyle \frac{7}{6}$};
\draw node[anchor = north] at (4,5) {$2$};

\fill[fill = red!20] (0, 5.86) rectangle (2.33, 7) ;
\draw[red, very thick] (0, 7) -- (2.33, 7) ;
\draw[red, very thick] (2.33, 7) -- (3,7); 
\draw node[anchor = south east] at (0, 7) {$1$};

\fill[fill = red!20] (2.33, 5) rectangle (3, 7) ;
draw[red, very thick] (3, 7) -- (3, 5) ;
\draw node[anchor = north] at (3,5) {$\displaystyle \frac{3}{2}$};

\matrix[draw, below left] at (current bounding box.north east)
{\node at ++(-2,0.5) [rectangle, draw = blue, fill= blue!20] {};   \node at ++(-0.5,0.65) [] {Agent 1};\\
\node at ++(-2, 0.5) [rectangle, draw = red, fill = red!20]{};
\node at ++(-0.5, 0.65) [] {Agent 2};\\
} ;

\end{tikzpicture}}
\caption{LCP Allocation described in Example \ref{lcp-example}}
\end{figure}

Before characterizing the behavior of LCP, we observe that it satisfies Pareto optimality and sharing incentives.

\begin{theorem} \label{LCP-po}
The LCP mechanism satisfies PO.
\end{theorem}

\begin{proof}
Suppose for contradiction the allocation output $\mathbf{A}^{\textnormal{LCP}}(\cdot)$ from the LCP mechanism is not PO. Then, there exists an alternative allocation $\mathbf{A}^{'}(\cdot)$ that decreases the cost of at least one agent without increasing the cost of any other agent. This implies $\underset{i \in \mathcal{N}}{\displaystyle \prod}  c_{i}(\mathbf{A}^{'}_{i}(\cdot)) < \underset{i \in \mathcal{N}}{\displaystyle \prod}  c_{i}(\mathbf{A}^{\textnormal{LCP}}_{i}(\cdot))$ which contradicts the optimality of the LCP solution.
\end{proof}

\begin{theorem}
\label{thm:SI}
The LCP mechanism satisfies SI.
\end{theorem}

\begin{proof}
Suppose for contradiction the LCP allocation violates SI for some agent $i$ who finishes at time $t_i > n k_i$. 
Let $\ell$ be the maximum value of $j$ such that $t_i - t_j > (n - j)k_i$.  (This is well defined as $j=0$ satisfies it).
Consider instead the alternative allocation where after the $\ell$th agent finishes we insert a period of length $k_i$ where agent $i$ has all the resources, which guarantees that agent $i$ finishes no later than $t_\ell + k_i$ but agents $j > \ell$ other than $i$ finish at time $t_j + k_i$ instead of $t_j$.  We show that this allocation improves the cost product, yielding the desired contradiction.

Taking the logarithm of the objective, the effect of making agent $j$ finish later is to increase it by $\log(t_j+k_i) - \log(t_j) = \int_{t_j}^{t_j+k_i} t^{-1} dt$.  The total increase is then
$\sum_{j > \ell, j \neq i} \int_{t_j}^{t_j+k_i} t^{-1} dt$.
The effect of making agent $i$ finish earlier is to decrease it by $\log(t_i) - \log(t_\ell + k_i) = \int_{t_\ell + k_i}^{t_i} t^{-1} dt > \int_{t_\ell + k_i}^{t_\ell+(n-\ell)k_i} t^{-1} dt = \sum_{\delta = 1}^{n-\ell-1}\int_{t_\ell + \delta k_i}^{t_\ell+(\delta+1)k_i} t^{-1} dt$, where the inequality follows from the definition of $\ell$.  Both sums have $n-\ell-1$ terms so we show that the lower bounds of the integration in the sum for $i$ are always strictly lower than those in the sum for the $j\neq i$, which in turn shows that the reduction in the log objective from $i$ completing earlier outweighs the increase from the $j \neq i$ completing later.

If $j > i$, then $t_j \geq t_i > t_\ell + \delta k_i$ as desired regardless of which $\delta$ corresponds to $j$.  If $i > j > \ell$, then since $j$ was not chosen as the value of $\ell$, $t_j - t_\ell > (j-\ell)k_i$.  Rewriting yields $t_j > t_\ell + (j-\ell)k_i$ as desired.
\end{proof}

\subsection{Structure of LCP Allocations}

In this subsection, we characterize the structure of LCP allocations. We show that, apart from the possibility of ties, they consist of allocations where on each time interval the allocation is an extreme point of the Pareto frontier.  For most instances, this suffices to reduce finding the exact allocation to examining a finite number of cases.  This characterization is used for several of our subsequent results and forms the basis of our simulations.

Let $\mathcal{N} = \{1,2,...,n\}$ denote the set of agents where the agents are numbered as per their completion times under the LCP mechanism, that is, agent 1 finishes first, agent 2 second, and so on. The cost product of the agents under the LCP mechanism is expressed as,

\begin{linenomath}
\begin{equation*}
    CP(\mathbf{A}^{\textnormal{LCP}}(\cdot)) = \displaystyle \prod_{i \in \mathcal{N}} t_{i}
\end{equation*}
\end{linenomath}

\noindent where $t_{i}$ denotes the completion time of agent $i \in \mathcal{N}$. Let $\lambda_{i,j}$ denote the dominant share allocated to agent \textit{i} during the time interval where agent \textit{j} completes its work. The completion time of agent 1 is given by,
$t_{1} = \displaystyle \frac{k_{1}}{\lambda_{1,1}}$, and for agents $i > 1$ the completion time is expressed recursively as,
\begin{linenomath}
\begin{equation*}
    t_{i} = \displaystyle \frac{k_{i} - \displaystyle \sum^{i -1}_{j = 1} ((t_{j} - t_{j-1}) \cdot \lambda_{i,j})}{\lambda_{i,i}} + t_{i-1} \quad \textnormal{(1)}.
\end{equation*}
\end{linenomath}

\noindent For example, the completion time of agent 2 is given by,
\begin{linenomath}
\begin{equation*}
    t_{2} = \displaystyle \frac{k_{2} - t_{1} \cdot \lambda_{2,1}}{\lambda_{2,2}} + t_{1} = \; \frac{k_{2} - \displaystyle \frac{k_{1}}{\lambda_{1,1}} \cdot \lambda_{2,1}}{\lambda_{2,2}} + \frac{k_{1}}{\lambda_{1,1}}.
\end{equation*}
\end{linenomath}

\noindent Now, we analyze the cost product of the LCP mechanism $CP(\mathbf{A}^{\textnormal{LCP}}(\cdot))$ by studying the cost product as a function of the parameters of the \textit{i}-th time interval, $CP_i(\lambda_{i,i}, \; \lambda_{i+1,i}, \; \lambda_{i+2, i}, \; \cdots, \lambda_{n,i})$ while taking the other allocation parameters as constants.
That is, $CP_i(\lambda) = CP(\mathbf{A}^\lambda)$, where $\mathbf{A}^\lambda$ is the modified LCP allocation with the $\lambda$ parameters for the $i$-th time interval.

\begin{lemma} \label{cost-prod1}
$CP_i(\lambda_{i,i},\; \lambda_{i+1,i}, \; \lambda_{i+2, i}, \; \cdots, \lambda_{n,i})$ is a Quasiconcave function.
\end{lemma}

\begin{proof}
First, we show that the product of the completion time of an agent with the dominant share allocated to the \textit{i}-th completing agent is affine. Then, we apply this property to transform the cost product to a univariate function. Finally, we show this univariate function is Quasiconcave.

Without loss of generality, when we analyze $CP_i(\lambda_{i,i},\; \lambda_{i+1,i}, \; \lambda_{i+2, i}, \; \cdots, \lambda_{n,i})$ we assume that agents $\{1, 2, , \cdots, i -1\}$ have completed their tasks resulting in completion times of agents $t_{1}, t_{2}, \cdots, t_{i-1}$ to be regarded as constants. In the $i$-th time interval, an updated tasks matrix is obtained as $\mathbf{K}^{'}$ where $k^{'}_{j} = 0$ for agents $1 \leq j \leq i-1$ since the tasks for set of agents $\{1, 2, \cdots, i-1\}$ have already been completed, and $0 \leq k^{'}_{j} \leq k_{j}$ for agents $i \leq j\leq n $.
That is, agents in $\mathcal{N}^{\;'} = \{i, i+1, \cdots, n\}$ have not completed their tasks but may have made some progress in earlier intervals.

Let $u_{j} = t_{j} \cdot \lambda_{i,i}$. We show by induction that $u_{j}$ is affine. In the base case we have $j = i$, for which $u_{i} =  t_{i} \lambda_{i,i} \; = k^{'}_i + \lambda_{i,i}t_{i-1}$. This is affine as $t_{i-1}$ is a constant.
Assume $u_{i}, u_{i+1}, \cdots, u_{j-1}$ are affine functions of $\lambda_{i,i}; \; \lambda_{i,i}, \lambda_{i+1,i}; \; \cdots; \; \lambda_{i,i}, ...,  \lambda_{j-1,i}$ respectively. We now show that $u_{j}$ is an affine function of $\lambda_{i,i}, \lambda_{i+1,i}, ..., \lambda_{j-1,i}, \lambda_{j,i}$.

\begin{linenomath}
\begin{align*}
\begin{split}
u_{j} & = t_{j} \cdot \lambda_{i,i} \\
& = \displaystyle \left( \frac{k^{'}_{j} - \displaystyle \sum_{a = i}^{j-1} (t_{a} - t_{a-1}) \cdot \lambda_{j,a}}{\lambda_{j,j}} + t_{j-1} \right) \cdot \lambda_{i,i}\\
& = \displaystyle \frac{k^{'}_{j}}{\lambda_{j,j}} \lambda_{i,i} - \frac{1}{\lambda_{j,j}} \displaystyle \sum_{a = i}^{j-1} (u_{a} - u_{a-1})\cdot \lambda_{j,a} + u_{j-1} \\
& = \displaystyle \frac{k^{'}_{j}}{\lambda_{j,j}} \lambda_{i,i} - \frac{k^{'}_{i}}{\lambda_{j,j}} \lambda_{j, i} - \frac{\displaystyle \sum_{a = i + 1}^{j-1} (u_{a} - u_{a-1}) \lambda_{j,a}}{\lambda_{j,j}}  + u_{j-1} \\
\end{split}
\end{align*}
\end{linenomath}

\noindent The first term is linear with respect to $\lambda_{i,i}$. The second term is linear with respect to $\lambda_{j,i}$. The last term of the expression is $u_{j-1}$ which is affine with respect to $\lambda_{i, i}, \lambda_{i+ 1,i}, \cdots , \lambda_{j-1, i}$. The expression in the third term, results in a linear combination of $u_{i}, u_{i+1}, \cdots, u_{j-1}$. Hence, $u_{j}$ is an affine function of $\lambda_{i,i}, \lambda_{i+1,i}, \cdots ,\lambda_{j,i}$. 

Let us call $\boldsymbol{{\lambda}} = (\lambda_{i,i}, \lambda_{i+1,i}, ... ,\lambda_{n,i})$ a valid allocation of $n - i + 1$ agents in the $i$-th time interval if $\lambda_{j,i} \geq 0$ for all $1 \leq i \leq j \leq n$, and $ \displaystyle \sum_{j \in \mathcal{N}^{'}} \lambda_{j,i} \cdot d_{j,r} \leq 1$ for all $r \in \mathcal{R}$. Let $\boldsymbol{\lambda}$ and $\boldsymbol{\lambda}^{'}$ represent any two valid allocations of the $n - i+1 $ agents in the $i$-th time interval. We transform the cost product during the $i$-th time interval into a univariate function $f_{\textnormal{CP}_{i}}(\theta)$ for $0 \leq \theta \leq 1$. Let,
\\$f_{\textnormal{CP}_{i}}(\theta) = CP_i(\theta \cdot \boldsymbol{\lambda} + (1 - \theta) \cdot \boldsymbol{\lambda}^{'})\quad$; or \\[2 pt]
$f_{\textnormal{CP}_{i}}(\theta) =  \displaystyle \prod_{j \in \mathcal{N}^{'}} \displaystyle \frac{u_{j}(\theta \cdot \boldsymbol{\lambda} + (1 - \theta) \cdot \boldsymbol{\lambda}^{'})}{(\theta \cdot \lambda_{i,i} + (1 - \theta) \cdot \lambda^{'}_{i,i}) }$ \\

\noindent Now, $CP_i(\boldsymbol{\lambda})$ is Quasiconcave if and only if the univariate function $f_{\textnormal{CP}_{i}}(\theta)$ is Quasiconcave for all such $\boldsymbol{\lambda}$ and $\boldsymbol{\lambda}^{'}$. Since, $u_{j}$ is an affine function of $\boldsymbol{\lambda}$,

\begin{linenomath}
\begin{equation*} \label{cp_univariate_eq1}
\begin{split}
f_{\textnormal{CP}_{i}}(\theta) & = \displaystyle \prod_{j \in \mathcal{N}^{'}} \displaystyle \frac{ \theta u_{j}(\boldsymbol{\lambda}) +  u_{j}(\boldsymbol{\lambda}^{'}) - \theta u_{j}(\boldsymbol{\lambda}^{'})}{\theta \lambda_{i,i} +  \lambda^{'}_{i,i} - \theta \lambda^{'}_{i,i}} \\
& = \displaystyle \prod_{j \in \mathcal{N}^{'}} \displaystyle \frac{\theta a_{j} + b_{j}}{\theta a_{0} + b_{0}} \quad \textnormal{(2)}
\end{split}
\end{equation*}
\end{linenomath}

\noindent where the following substitutions are made at Equation (2), $a_{j} = u_{j}(\boldsymbol{\lambda}) - u_{j}(\boldsymbol{\lambda}^{'})\;$, $b_{j} = u_{j}(\boldsymbol{\lambda}^{'})\;$, $a_{0} = \lambda_{i,i} - \lambda^{'}_{i,i}\;$, and $b_{0} = \lambda^{'}_{i,i}$. By construction, $b_{j}, b_{0}\geq 0$. Without loss of generality, we assume $a_{j} \geq 0$.

We analyze the logarithm of $f_{\textnormal{CP}_{i}}(\theta)$ instead of $f_{\textnormal{CP}_{i}}(\theta)$ since Quasiconcavity is preserved under monotonic transformations. Lemma \ref{cost-prod2} below shows $\log(f_{\textnormal{CP}_{i}}(\theta))$ is Quasiconcave, so $CP_{i}(\boldsymbol{\lambda})$ is a Quasiconcave function.
\end{proof}

\begin{lemma} \label{cost-prod2}
Let $f(\theta) = \displaystyle \prod_{j \in \mathcal{N}^{'}} \displaystyle \frac{\theta a_{j} + b_{j}}{\theta a_{0} + b_{0}}$.  Suppose $f(\theta) > 0$ for $\theta \in [0,1]$; $a_{0}, b_{0} \geq 0$; and $b_{j} \geq 0$ for $j \in \mathcal{N}^{'}\subseteq \mathcal{N}$. Then, $\log(f(\theta))$ is Quasiconcave on $[0,1]$.
\end{lemma}

We will show that the set of allocations, considered for each period is convex. Since the minimum of a function on a convex set is at one of the extreme points, Lemma \ref{cost-prod1} tells us we can restrict to examining these, of which there are a finite number (holding the parameters for other periods fixed).
We know that LCP allocations are Pareto optimal from Theorem \ref{LCP-po}. Hence, we must look at the extreme points that lie on the Pareto frontier for every combination from the set of $\mathcal{N}$ agents. The Pareto frontier is determined by the capacity of a resource, $ \sum_{i \in \mathcal{N}} \lambda_{i} \cdot d_{ir} \; \leq \; 1, \quad \forall r \in \mathcal{R}$
where $\lambda_{i}$ denotes the instantaneous allocation of agent \textit{i}, and $d_{ir}$ corresponds to the demand of agent \textit{i} for resource \textit{r}. The Pareto frontier is the boundary defined by at most \textit{m} hyperplanes where each hyperplane corresponds to a resource $r \in \mathcal{R}$ getting saturated. The extreme points are defined by the relevant intersections of these hyperplanes, i.e. points where at least two resources are saturated as well as points where a single agent is allocated. This characterization is similar in spirit to the result obtained by \citeauthor{bogomolnaia2017competitive} \shortcite{bogomolnaia2017competitive} that an allocation is \textit{competitive} when dividing a set of bads if and only if the utility profile is negative and the allocation is a critical point of the product of the absolute values of the utilities in the negative efficiency frontier. 

However, our setting has one additional complication not present in their setting.  In addition to resource constraints, since our optimization is based on optimizing on one period while holding the others constant we also have constraints to ensure that all agents finish in the correct order. There may be extreme points of the feasible set where these constraints are tight, which represent ties. We give such an example in the Appendix.

\begin{theorem} \label{cost-prodLCP}
An LCP allocation $\mathbf{A}^{\textnormal{LCP}}(\cdot)$ consists of at most $n$ time interval allocations.  Unless two agents tie in completion time, the allocation of agents in each interval is an extreme point of the Pareto frontier.
\end{theorem}

\begin{proof}
Let the LCP solution be given, with agents ordered by their completion time.  Since this minimizes the cost product, it must also do so when we only optimize over the allocation in one time period, holding all the others fixed.

We know from Lemma \ref{cost-prod1} that the cost product function at the $i$-th time interval expressed as a function of the allocation of $n - i + 1$ agents left to finish their work $CP_i(\lambda_{i,i}, \; \lambda_{i+1,i}, \; \cdots ,\; \lambda_{n,i})$ is Quasiconcave. The minimum of a Quasiconcave function over a closed convex set is attained at the extreme points of that closed convex set.

Our optimization problem over $(\lambda_{i,i}, \; \lambda_{i+1,i}, \; \cdots ,\; \lambda_{n,i})$ has three types of constraints.  The first two, $\lambda_{j,i} \geq 0$ and $\sum_{j \in \mathcal{N}} \lambda_{j} \cdot d_{jr} \leq 1$ are non-strict linear inequalities, so their intersection is a closed convex set.  The third ensures that agent $j$ is in fact the $j$-th agent to finish, i.e., $t_j \leq t_{j+1}$ for all $j$.  As shown in the proof of Lemma~\ref{cost-prod1}, the $u_j$ (which are a monotone transformation of the $t_j$) are affine functions, so each of these constraints defines a closed convex set.  Our feasible set represents the intersection of these constraints as well and hence, it is closed and convex.

We also know from Theorem \ref{LCP-po} that the LCP allocations are Pareto optimal. Thus, we only need to consider extreme points which lie on the Pareto frontier.  If these extreme points are defined only by the first two types of constraints, then they are extreme points of the Pareto frontier itself.  Otherwise, they include a constraint of the third type being tight, i.e. there are two agents with a tie in completion time.
\end{proof}

\begin{lemma} \label{lcp-sp}
The LCP mechanism violates SP.
\end{lemma}

\section{Envy-freeness of LCP}

In all prior settings we are aware of with divisible goods, solutions based on maximizing Nash welfare have been envy-free \cite{VARIAN197463, moulin2004fair, bogomolnaia2017competitive,bogomolnaia2019dividing,goel2019markets}. Even in non-divisible settings, it has strong envy-freeness properties that have led it to be described as "unreasonably fair" \cite{Caragiannis2016MaxNashWelfare}.  Surprisingly, we show in our setting that the LCP solution is not envy-free, even in expectation.

\begin{lemma} \label{lcp-envy}
The LCP mechanism does not satisfy EF (even in expectation) when the resource allocation involves three or more agents and two or more resources.  
\end{lemma}


Consider an example of three agents requiring two resources that have the demand matrix $\mathbf{D}$, and $\mathbf{K}$ as follows, \\[2 pt]
\noindent $\mathbf{D} = \begin{bmatrix} 1 \quad & 1 \\
1 \quad & \epsilon \\
\epsilon \quad & 1\end{bmatrix} ,\;$ 
\noindent $\mathbf{K} = \begin{bmatrix} 1 \quad & 0 \quad & 0 \\
0 \quad & 1 \quad & 0 \\
0 \quad & 0 \quad & k_{3} \\\end{bmatrix}, \\$ 

\noindent where $\epsilon << 1$ and $k_{3} > 2$. 

Figure \ref{lcp-envy-eg}, shows the form of the LCP allocation for the above demands $\mathbf{D}$, and $\mathbf{K}$. In this example, the first and second agents have the same dominant resource and are required to complete the same amount of work. Also, the first agent demands more of the second resource compared to the second agent while, the third agent has a different dominant resource when compared to the second agent. As illustrated by Figure \ref{lcp-envy-eg}, the LCP mechanism would be lowering the cost product by first completing the allocation of the first agent and then allowing the second and third agents to share their allocations. This allocation is efficient but the second agent would envy the first agent. A more detailed description of the example is given in the Appendix.  

\begin{figure}[!ht] 
 \scalebox{0.68}{\begin{tikzpicture}
\draw  node[anchor = west] at (-2, 1) {$\text{R1}$} ;
\draw[very thick, ->] (-0.2, 0) -- (9.6, 0) node[right] {$\text{time} \quad  t$} ;
\draw[very thick, ->] (0, -0.2) -- (0, 3.2) node[above] {$\text{Allocation} \quad \mathbf{A}_{i1}(t)$} ;
\draw node[anchor = north west] at (0, 0) {$\text{0}$}; 
\draw node[anchor =  east] at (-0.1, 0) {$\text{0}$};

\pgfplothandlermark{\pgfuseplotmark{|}}
\pgfplotstreamstart
\pgfplotstreampoint{\pgfpoint{0cm}{0cm}}
\pgfplotstreampoint{\pgfpoint{2cm}{0cm}}
\pgfplotstreampoint{\pgfpoint{4.2cm}{0cm}}
\pgfplotstreampoint{\pgfpoint{8.2cm}{0cm}}
\pgfplotstreamend

\pgfplothandlermark{\pgfuseplotmark{-}}
\pgfplotstreamstart
\pgfplotstreampoint{\pgfpoint{0cm}{0cm}}
\pgfplotstreampoint{\pgfpoint{0cm}{0.2cm}}
\pgfplotstreampoint{\pgfpoint{0cm}{1.82cm}}
\pgfplotstreampoint{\pgfpoint{0cm}{2cm}}
\pgfplotstreamend

\fill[fill= blue!20] (0, 0) rectangle (2, 2) ;
\draw[blue, thick] (0, 2) -- (2,2);
\draw[blue, thick] (2,2) -- (2,0);
\draw node[anchor = north] at (2, 0) {$\text{1}$}; 
\draw node[anchor = east] at (0, 2) {$\text{1}$} ;

\fill[fill = green!20] (2, 0) rectangle (4.2, 1.82) ;
\draw[green, thick] (2, 1.82) -- (4.2, 1.82) ;
\draw[green, thick] (4.2, 1.82) -- (4.2, 0); 
\draw node[anchor = north] at (4.2, 0) {$2+ \epsilon$};
\draw node[anchor = north east] at (0, 1.82) {$\displaystyle \frac{1}{1+\epsilon}$} ;

\fill[fill = red!20] (2, 1.82) rectangle (4.2, 2) ;
\draw[red, thick] (2, 2) -- (4.2, 2) ;
\draw[red, thick] (4.2, 2) -- (4.2, 1.82) ;
\fill[fill = red!20] (4.2, 0) rectangle (8.2, 0.2);
\draw[red, thick] (4.2, 0.2) -- (8.2, 0.2);
\draw[red, thick] (8.2, 0.2) -- (8.2, 0);
\draw node[anchor = north] at (8.2, 0) {$1+ \epsilon + k_{3} $};
\draw node[anchor = east] at (0, 0.2) {$\epsilon$} ;

\draw  node[anchor = west] at (-2, 6) {$\text{R2}$} ;
\draw[very thick, ->] (-0.2, 5) -- (9.6, 5) node[right] {$\text{time} \quad  t$} ;
\draw[very thick, ->] (0, 4.8) -- (0, 8.2) node[above] {$\text{Allocation} \quad \mathbf{A}_{i2}(t)$} ;
\draw node[anchor = north west] at (0,5) {$0$};
\draw node[anchor = east] at (-0.1,5) {$0$} ;

\pgfplothandlermark{\pgfuseplotmark{|}}
\pgfplotstreamstart
\pgfplotstreampoint{\pgfpoint{0cm}{5cm}}
\pgfplotstreampoint{\pgfpoint{2cm}{5cm}}
\pgfplotstreampoint{\pgfpoint{4.2cm}{5cm}}
\pgfplotstreampoint{\pgfpoint{8.2cm}{5cm}}
\pgfplotstreamend

\pgfplothandlermark{\pgfuseplotmark{-}}
\pgfplotstreamstart
\pgfplotstreampoint{\pgfpoint{0cm}{5cm}}
\pgfplotstreampoint{\pgfpoint{0cm}{5.18cm}}
\pgfplotstreampoint{\pgfpoint{0cm}{7cm}}
\pgfplotstreamend

\fill[fill= blue!20] (0, 5) rectangle (2, 7) ;
\draw[blue, thick] (0, 7) -- (2,7);
\draw[blue, thick] (2,7) -- (2,5);
\draw node[anchor = north] at (2,5) {$1$};
\draw node[anchor = east] at (0,7) {$1$};

\fill[fill = green!20] (2, 5) rectangle (4.2, 5.18) ;
\draw[green, thick] (2, 5.18) -- (4.2, 5.18) ;
\draw[green, thick] (4.2, 5.18) -- (4.2, 5); 
\draw node[anchor = north] at (4.2, 5) {$2+ \epsilon$} ;
\draw node[anchor = south east] at (0, 5.18) {$\displaystyle \frac{\epsilon}{1 + \epsilon}$};

\fill[fill = red!20] (2, 5.18) rectangle (4.2, 7) ;
\draw[red, thick] (2, 7) -- (4.2, 7) ;
\draw[red, thick] (4.2, 2) -- (4.2, 1.82) ;
\fill[fill = red!20] (4.2, 5) rectangle (8.2, 7);
\draw[red, thick] (4.2, 7) -- (8.2, 7);
\draw[red, thick] (8.2, 7) -- (8.2, 5);
\draw node[anchor = north] at (8.2, 5) {$1+ \epsilon + k_{3}$};

\matrix[draw, below left] at (current bounding box.north east)
{\node at ++(-2,0.5) [rectangle, draw = blue, fill= blue!20] {};   \node at ++(-0.5,0.65) [] {Agent 1};\\
\node at ++(-2, 0.5) [rectangle, draw = green, fill = green!20]{};
\node at ++(-0.5, 0.65) [] {Agent 2};\\
\node at ++(-2, 0.5) [rectangle, draw = red, fill = red!20]{};
\node at ++(-0.5, 0.65) [] {Agent 3};\\
} ;

\end{tikzpicture}}
\caption{Example where LCP has envy} \label{lcp-envy-eg}
\end{figure}
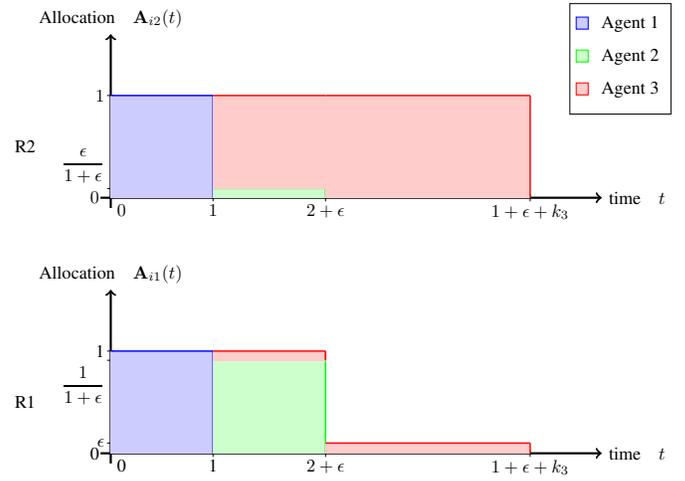

While these examples rule out envy-freeness in general, they seem to be relatively rare in randomly generated examples as our simulations show.

The requirement of three agents and two goods is tight.  We show in the Appendix that with only two agents or a single resource the LCP allocation is EF in expectation.

\begin{lemma} \label{lcp-1res_and_2agents}
For single resource allocations, the LCP mechanism allocates the resource in the form of the Shortest Job First (SJF) and the allocations are envy-free in expectation. The LCP allocations are envy-free in expectation when allocating multiple resources for two agents.
\end{lemma}

\section{Simulation Results} \label{simulation}

\begin{table*}[t]
\centering
\begin{tabular}{|p{32mm}| p{32mm}| p{15mm}| p{15mm}| p{15mm}|p{15mm}|}
\hline
Evaluation metrics & Mechanism & \multicolumn{4}{c|}{Number of agents $n$}\\
\cline{3-6}
&& $n = 2$ & $n = 3$ & $n = 4$ & $n = 5$\\
\hline
Envy-freeness & LCP-X &  100 \% &  99.95 \% &  99.95 \% &  99.90 \%\\
\cline{2-6}
& DRF-W  &  100 \% &  100\% &  100 \% &  100\%\\
\hline
Sharing-Incentives & LCP-X &  100 \% &  100\% &  100 \% &  100 \%\\
\cline{2-6}
& DRF-W  &  100 \% &  100 \% &  100 \% &  100 \%\\
\hline
Lower & LCP-X &  2.1 \% &  3.6 \% &  4.7 \% &  5.1 \%\\
\cline{2-6}
Makespan & DRF-W  &  58.3 \% &  73.6 \% &  76 \% &  78.6 \%\\
\cline{2-6}
& LCP-X equals DRF-W  &  39.6 \% &  23.4 \% &  19.3 \% &  16.3 \% \\
\hline
Lower mean & LCP-X &  95.65 \% &  99.3 \% &  100 \% &  100 \%\\
\cline{2-6}
completion time & DRF-W &  4.35 \% &  0.7 \% &  0 \% &  0 \%\\
\hline
\multicolumn{2}{|c|}{LCP-X Pareto dominates DRF-W} &  39.6 \% &  24 \% &  20.3 \% &   18.15 \% \\
\hline
\end{tabular}
\caption{\label{tab:1}Comparison of LCP-X and DRF-W on 2000 instances for each $n$}
\end{table*}

We simulate the LCP and DRF-W mechanisms on randomly generated problem instances in order to better understand the trade-off between fairness and efficiency. The simulation varies the number of agents from $2$ to $5$. For each number of agents, 2000 examples were generated. The number of resources for each example was chosen uniformly at random between 1 to 10. The demand vector of an agent was generated using a uniform distribution on $(0.0, \; 1.0]$. The demand vector was then normalized for each agent. The amount of work $k_{i}$ required for agent \textit{i} to complete was chosen uniformly at random from $(0.0, \; 100.0]$. Using uniform randomness tends to make it harder for the LCP mechanism when compared to more realistic distributions, since if some resources are globally more desired than others or if there are large disparities in $k_{i}$ then there are effectively fewer resources or agents to consider.

The DRF-W allocation is straightforward to compute as a DRF allocation among agents over each interval. Theorem~ \ref{cost-prodLCP}, characterizes the structure of LCP allocations, but the possibility of ties makes an exact computation challenging. Instead, we consider only the main branch of the characterization, where in each period the allocation is an extreme point of the Pareto frontier. We refer to the algorithm that considers only such candidates as LCP-X. With this restriction, we can compute the LCP-X allocation by enumeration, which is polynomial in the number of resources but exponential in the number of agents. While LCP-X does not immediately inherit the properties of LCP, the LCP-X allocation is based on the same principle as the LCP mechanism and is substantially easier to compute because the LCP-X minimizes over a restricted set that excludes certain types of ties. Our simulations show that despite this restriction,  the LCP-X generally satisfies the same desiderata as the LCP.


Table \ref{tab:1} summarizes the results of the comparisons between the DRF-W allocation and the LCP-X allocation. Consistent with our theoretical results for LCP, the LCP-X allocations always satisfied SI and were EF with two agents. Rare instances showed envy for LCP-X with more than two agents, but upon inspection these were instances where LCP would have envy as well.  This suggests that, at least for randomly generated instances, we should expect the LCP mechanism to almost always satisfy EF.  Our use of a continuous distribution for the amount of work each agent has ensured that ties in $k_i$ did not occur. Thus, there were no instances that would have satisfied EF in expectation but not EF.

We also compare DRF-W and LCP-X in terms of their makespan and average completion time, metrics commonly used in the systems settings which inspired our work \cite{Grandl2016Carbyne}. The makespan is the completion time of the last agent to complete its work and can be interpreted as a fairness property. We observe that the DRF-W allocation typically has a lower makespan than the LCP-X allocation. Since each agent makes equal progress in the DRF-W allocation, with more agents the makespan of the DRF-W mechanism decreases at a faster rate compared to the LCP mechanism. In the instances where the LCP-X allocation had a lower makespan, the LCP-X shares more aggressively, slightly slowing down the initial agents relative to the DRF-W allocation while speeding up the later agents.

The mean completion time of an agent represents an estimate of how fast an agent is expected to complete its work under the allocation and is an efficiency property. For most cases, the LCP-X allocation has a lower mean completion time because of its higher efficiency. In the examples where the DRF-W allocation had a lower mean completion time, we found that the LCP-X allocation takes the form of the shortest job first allocation whereas the sharing of the DRF-W allocation does a better job at minimizing the makespan thereby lowering the mean completion time.

Finally, we examined if the LCP-X allocation Pareto dominated the DRF-W allocation. From Table \ref{tab:1}, we see that the percentage of instances where the LCP-X allocation Pareto dominates the DRF-W allocation  decreases as the number of agents increase from $n = 2$ to $n = 5$. This decrease is largely accounted for the increase in the percentage of instances the DRF-W allocation has a lower makespan than the LCP-X allocation.

\section{Conclusion}

We analyzed the fair division problem for agents having Leontief preferences with limited demands. We showed that DRF-W allocations exhibit several fairness properties but are in general inefficient. The LCP mechanism on the other hand, outputs allocations that are highly efficient but have relatively weak fairness properties. For two agents, we see that the LCP allocations are better as they are envy-free in expectation and Pareto efficient. For more than two agents, the LCP mechanism does not satisfy envy-freeness (even in expectation). However, the LCP allocations satisfy the weaker fairness property of sharing incentives and the violations of envy-freeness rarely occurred in simulations.
 
We believe this is a reasonable trade-off because when applying our model in system settings (for example as in \citet{Grandl2016Carbyne}, \citet{Hindman2011Mesos}, \citet{vavilapalli2013apacheHadoopYARN}), the primary objection to DRF is its lack of efficiency. In contrast, LCP yields substantial efficiency gains and in simulations still appears to generally be fair despite the existence of non-EF examples. Also, we do not address the computational complexity of the LCP mechanism. However, we present LCP-X in the previous section \ref{simulation} which is a simplified version of the LCP mechanism. Our naive implementation of the LCP-X is polynomial in terms of the number of resources but exponential in the number of agents.

For future work, our observation that there is room to improve the efficiency of DRF at essentially no cost to fairness was motivated by \citet{Grandl2016Carbyne}, who heuristically apply this to jobs with online arrivals and jobs composed of subtasks with a DAG structure.  We plan to explore the application of the LCP solution in this richer setting.

Another direction would be to investigate a market interpretation of the LCP mechanism similar to the works of \citet{goel2019markets}, and the Eisenberg-Gale convex program. One issue here is that usually EF can be shown from the convex program but we know that the LCP mechanism does not satisfy EF.

\bibliography{bibfile}

\clearpage
\newpage

\section{Appendix}

This portion of the Appendix contains the proofs and other material omitted from the main manuscript.

\subsection{Preliminaries}

\begin{customlemma}{1}
Given any resource allocation $\mathbf{A}(\cdot)$, there exists an equivalent fixed allocation $\mathbf{A}^{\textnormal{fixed}}(\cdot)$ where all the agents have the same costs / utilities.
\end{customlemma}

\begin{proof}
Let us assume $\mathbf{A}(\cdot)$ to be the given resource allocation with completion times as $\{t_{i}, \quad \forall i \in \mathcal{N}\}$.
We show how to construct a $\mathbf{A}^{\textnormal{fixed}}(\cdot)$ to be the equivalent fixed allocation defined over time intervals determined by the completion times of $\mathbf{A}(\cdot)$.

The fixed allocation of an agent over an interval is determined by the average amount of work done by the agent over that interval. Let $\Delta w_{i}$ represent the amount of work completed by an agent \textit{i} over the interval $[t_{a}, \; t_{b})$ for $\Delta t_{b} = t_{b} - t_{a}$ units of time. Then, $\Delta w_{i}$ can expressed as $\Delta w_{i}  = \int_{t_{a}}^{t_{b}} \mathbf{A}_{i}(t) dt$. The fixed allocation of agent \textit{i} over the interval is expressed as follows \\$\mathbf{A}^{\textnormal{fixed}}_{i}(t) = \dfrac{\Delta w_{i}}{\Delta t_{b}}, \quad \forall t \in [t_{a}, \; t_{b})$.

We first show that $\mathbf{A}^{\textnormal{fixed}}(\cdot)$ generates an instantaneous allocation at all times. That is, 
\begin{enumerate}
    \item $A^{\textnormal{fixed}}_{ir}(t) = \lambda_{i} \cdot d_{ir}, \quad \forall i \in \mathcal{N}, \; \forall r \in \mathcal{R}, \; \forall t$ for some $\lambda_{i} \in \mathbb{R}_{\geq 0}$.
    
    \item $\displaystyle \sum_{i \in \mathcal{N}} A^{\textnormal{fixed}}_{ir}(t) \leq  1, \quad \forall r \in \mathcal{R}, \; \forall t$.
    
\end{enumerate}

From the definition of a fixed allocation, we know that during any time interval $[t_{a}, \; t_{b})$ an agent \textit{i} consumes a total of $\Delta w_{ir}$ of resource \textit{r}. The fixed allocation over this interval is given by,
\begin{linenomath}
\begin{align*}
\begin{split}
A^{\textnormal{fixed}}_{ir}(t) & = \dfrac{\Delta w_{ir}}{\Delta t_{b}} 
\; = \; \dfrac{\int_{t_{a}}^{t_{b}} A_{ir}(t) dt}{\Delta t_{b}} \; \\ & = \dfrac{\int_{t_{a}}^{t_{b}} \lambda_{i}(t) \cdot d_{ir} \; dt}{\Delta t_{b}} \; = \;
\dfrac{\int_{t_{a}}^{t_{b}} \lambda_{i}(t) \; dt}{\Delta t_{b}} d_{ir} \\ & = \; \hat{\lambda}_{i} \cdot d_{ir}
\end{split}
\end{align*}
\end{linenomath}

\noindent where, $\hat{\lambda}_{i} = \int_{t_{a}}^{t_{b}} \lambda_{i}(t)dt / \Delta t_{b}$. Since $\int_{t_{a}}^{t_{b}} \lambda_{i}(t)dt \geq 0$, and $\Delta t_{b} > 0$, we have $\hat{\lambda}_{i} \geq 0,\; \forall t \in [t_{a}, \; t_{b}]$.

\noindent From the given allocation, we have $\displaystyle \sum_{i \in \mathcal{N}} A_{ir}(t) \leq 1, \quad \forall r \in \mathcal{R}, \; \forall t$. Consider the amount of resource \textit{r} consumed by all the agents in the time interval $[t_{a}, \; t_{b})$, that is, \\
$\displaystyle \int_{t_{a}}^{t_{b}} \sum_{i \in \mathcal{N}} A_{ir}(t) \; dt \; \leq \; \displaystyle \int_{t_{a}}^{t_{b}} 1 dt \; = \; t_{b} - t_{a}\quad$; or,
\\$\displaystyle \sum_{i \in \mathcal{N}} \int_{t_{a}}^{t_{b}} A_{ir}(t) \; dt \; \leq \; t_{b} - t_{a}$.
\\[6 pt]Substituting from the definition of $\Delta w_{i}$ yields,  \\
$\displaystyle \sum_{i \in \mathcal{N}} \Delta w_{ir} \; \leq \;  \Delta t_{b} \quad$; or,
\\$\displaystyle \sum_{i \in \mathcal{N}} \dfrac{\Delta w_{ir}}{\Delta t_{b}} \; \leq \; 1 \quad$; or,
\\$\displaystyle \sum_{i \in \mathcal{N}} A^{\textnormal{fixed}}_{ir}(t) \; \leq \; 1, \quad \forall t \in [t_{a}, \; t_{b})$.

\noindent Hence, the second condition for $\mathbf{A}^{\textnormal{fixed}}(\cdot)$ to be an instantaneous allocation is satisfied.

Given an allocation $\mathbf{A}(\cdot)$, let us assume an agent \textit{i} completes $\Delta w^{\Delta t_{b}}_{i}$ amount of work over the interval $[t_{a}, \; t_{b})$. The same amount of work gets completed by the fixed allocation $\mathbf{A}^{\textnormal{fixed}}(\cdot)$ as suggested by the following expression, 
$\Delta w^{\Delta t_{b}}_{i} \; = \; \int_{t_{a}}^{t_{b}} \mathbf{A}_{i}(t)\;dt \; =\;  A^{\Delta t_{b}}_{i} \cdot \Delta t_{b}$, where $\mathbf{A}^{\textnormal{fixed}}_{i}(t) \; = \; A^{\Delta t_{b}}_{i}$  is the fixed allocation of agent \textit{i} for $t \in [t_{a}, \; t_{b})$. Without loss of generality, let us assume the agents are numbered as per their completion times, that is, $t_{i}$ denotes the \textit{i}-th agent to complete its work. The first agent completes its work at $t_{1}$ for which,\\[2 pt] 
$\mathbf{w}_{1} \; = \; \displaystyle \int_{0}^{t_{1}} \mathbf{A}_{1}(t) dt\; = \; A^{\Delta t_{1}}_{1} \cdot t_{1} \quad \forall t \in [0, \;\; t_{1})$.\\[2 pt]

\noindent In general, the i-th agent completes it's work as per,\\[4 pt]
$\mathbf{w}_{i} \; = \; \displaystyle \int_{0}^{t_{i}} \mathbf{A}_{i}(t) dt \quad$; or, \\[4 pt]
$\mathbf{w}_{i} = \displaystyle \int_{0}^{t_{1}} \mathbf{A}_{i}(t) dt + \int_{t_{1}}^{t_{2}} \mathbf{A}_{i}(t) dt + ... +\int_{t_{i-1}}^{t_{i}} \mathbf{A}_{i}(t) dt$; or,\\[4 pt]     
$\mathbf{w}_{i} \; = \; A^{\Delta t_{1}}_{i} \cdot \Delta t_{1} + A^{\Delta t_{2}}_{i} \cdot \Delta t_{2} + ... + A^{\Delta t_{i}}_{i} \cdot \Delta t_{i}$.\\[4 pt]

\noindent From the above equations, we can infer that any agent \textit{i} in the fixed allocation $\mathbf{A}^{\textnormal{fixed}}(\cdot)$ is completing the same amount of work as the given allocation $\mathbf{A}(\cdot)$ by its completion time $t_{i}$, with the only difference being that the agents in the fixed allocation consume resources at constant rates during each interval. Hence, the agents in a given allocation $\mathbf{A}(\cdot)$ have the same costs/utilities as that in an equivalent fixed allocation $\mathbf{A}^{\textnormal{fixed}}(\cdot)$.
\end{proof}

\subsection{Dominant Resource Fairness with Work}

\begin{customthm}{1}
DRF-W mechanism satisfies weak-PO, SI, EF, and SP.
\end{customthm}

\begin{proof}
First, we show that the DRF-W mechanism satisfies SI. The linear program used to compute the DRF-W allocation at time \textit{t} can be expressed as follows,
\\$\lambda(t)  = \dfrac{1}{\underset{r \in \mathcal{R}}{\max} \displaystyle \sum_{i \in \mathcal{N}} d_{ir}}$,

\noindent where $\lambda(t)$ is the dominant share of the agents at time \textit{t}.
Since DRF satisfies SI, $\lambda_i \geq \dfrac{1}{n}$ at all $t$ for all agents $i$. The cost of an agent satisfies $c_{i}(\mathbf{A}_{i}(\cdot)) \leq n \cdot k_i$. 

We now show how the DRF-W mechanism satisfies weak Pareto-optimality. 
Call an allocation $\mathbf{A}(\cdot)$ a \textit{non-wasteful} allocation if the feasibility condition of the instantaneous allocation holds with an equality for at least one resource at all times. This implies that in a non-wasteful allocation there is at least one saturated resource at all times. The DRF-W mechanism is a non-wasteful allocation.
Since at least one resource is fully allocated during the first interval, making any agent achieve strictly more progress during this interval requires making some other agent achieve strictly less.  The result follows by induction.

Since DRF-satisfies SP, any alternate report made by agent $i$ will weakly decrease $\mathbf{A}_{i}(t)$ at all times $t$, thus ensuring that DRF-W satisfies SP as well.

Finally, we show that the DRF-W mechanism satisfies envy-freeness. Suppose there is an agent \textit{i} that envies the allocation of agent \textit{j}, that is, $c_{i}(\mathbf{A}_{j}(\cdot)) \leq c_{i}(\mathbf{A}_{i}(\cdot))$. This means agent \textit{i} would be able to make a better utilization of the allocation given to agent \textit{j}. This is only possible when agent \textit{j} has a strictly greater share of every resource than agent \textit{i}, that is, $\lambda_{j}(t) > \lambda_{i}(t)$. However, this would contradict the progressive filling allocation policy that the dominant resource of agents \textit{i} and \textit{j} is allocated at an equal rate, that is $\lambda_{i}(t) = \lambda_{j}(t) = \dfrac{1}{\underset{r \in \mathcal{R}}{\max} \displaystyle \sum_{i \in \mathcal{N}} d_{ir}}$. 

\end{proof}

\subsection{Least Cost Product Mechanism}

\begin{customlemma}{3}
Let $f(\theta) = \displaystyle \prod_{j \in \mathcal{N}^{'}} \displaystyle \frac{\theta a_{j} + b_{j}}{\theta a_{0} + b_{0}}$.  Suppose $f(\theta) > 0$ for $\theta \in [0,1]$; $a_{0}, b_{0} \geq 0$; and $b_{j} \geq 0$ for $j \in \mathcal{N}^{\;'} \subseteq \mathcal{N}$. Then $\log(f(\theta))$ is Quasiconcave on $[0,1]$.
\end{customlemma}

\begin{proof}
Let $g(\theta) = \log(f(\theta))$. We show that the function $g(\theta)$ is either non-increasing or non-decreasing or there is point $p \geq 0$ such that for $\theta \leq p ,\;$ $g(\theta)$ is non-decreasing, and for $\theta \geq p,\;$ $g(\theta)$ is non-increasing. That is, we show that 
\begin{enumerate}
    \item If $g^{'}(0) \leq 0$, then $g^{'}(\theta) \leq 0$ for all $\theta > 0$.
    
    \item If $g^{'}(0) \geq 0$, then either $g^{'}(\theta) \geq 0$ for all $\theta > 0$ or if $g^{'}(\theta)$ changes from positive to negative at some point $p \geq 0$, then $g^{'}(\theta) \leq 0$ for $\theta \geq p$.
\end{enumerate}

\noindent We simplify $g(\theta)$ as
\begin{linenomath}
\begin{equation*}
\begin{split}
g(\theta) & = \log\left(\displaystyle \prod_{j \in \mathcal{N}^{'}} \displaystyle \frac{\theta a_{j} + b_{j}}{\theta a_{0} + b_{0}} \right) = \displaystyle \sum_{j \in \mathcal{N}^{'}} \log\left(\displaystyle \frac{\theta a_{j} + b_{j}}{\theta a_{0} + b_{0}} \right)
\\[4 pt] & = \displaystyle \sum_{j \in \mathcal{N}^{'}} \left( \log(\theta a_{j} + b_{j}) - \log(\theta a_{0} + b_{0}) \right)
\\[4 pt] & = \displaystyle \sum_{j \in \mathcal{N}^{'}} \log(\theta a_{j} + b_{j}) - (n - i + 1) \log(\theta a_{0} +b_{0}) 
\end{split}
\end{equation*}
\end{linenomath}

\noindent Before considering the general case, let us examine $g(\theta)$ when $a_{0} = 0$. If $a_{0} = 0$, then $g(\theta)$ is expressed as a sum of log of affine functions as shown in Equation (1).
\begin{linenomath}
\begin{equation*}
g(\theta) = \displaystyle \sum_{j \in \mathcal{N}^{'}} \log(\theta a_{j} + b_{j}) - (n-i+1) \log(b_{0}) \quad \textnormal{(1)} 
\end{equation*}
\end{linenomath}

\noindent We know that the log of an affine function is concave and concavity is preserved under addition. Hence, $g(\theta)$ is concave and thus Quasiconcave.

\noindent Let $\mathcal{Z} \subseteq \mathcal{N}^{\;'}, \;$ $\mathcal{U} = \mathcal{N}^{\;'} \backslash \mathcal{Z}, \;$ $\zeta = |{\mathcal{Z}}|$, and $\;\overline{n} = (n - i + 1) - \zeta = |\mathcal{U}|$ such that $a_{j} = 0$ for $j \in \mathcal{Z}$, $a_{0} \neq 0$, and $a_{j} \neq 0$ for $j \in \mathcal{U}$. Using the chain rule of differentiation we have,
\begin{linenomath}
\begin{equation*}
\begin{split}
& \frac{d g(\theta)}{d \theta} = \displaystyle \sum_{j \in \mathcal{U}} \displaystyle \frac{a_{j}}{\theta a_{j} + b_{j}} - (n - i + 1) \displaystyle \frac{a_{0}}{\theta a_{0} + b_{0}} 
\\[4 pt] & = \displaystyle \sum_{j \in \mathcal{U}} \displaystyle \frac{1}{\theta + \displaystyle \frac{b_{j}}{a_{j}}} - (n - i + 1) \displaystyle \frac{1}{\theta + \displaystyle \frac{b_{0}}{a_{0}}}
\\[4 pt] & = \displaystyle \frac{\left(\theta + \displaystyle \frac{b_{0}}{a_{0}} \right) \displaystyle \sum_{j \in \mathcal{U}} \;\displaystyle \prod_{\ell \in \mathcal{U} \backslash j} \left( \theta + \frac{b_{\ell}}{a_{\ell}}\right) - (n - i + 1) \prod_{j \in \mathcal{U}} \left(\theta + \displaystyle \frac{b_{j}}{a_{j}} \right) }{\left( \theta + \displaystyle \frac{b_{0}}{a_{0}} \right) \displaystyle \prod_{j \in \mathcal{U}} \left(\theta + \displaystyle \frac{b_{j}}{a_{j}} \right)}
\\[4 pt] & = \displaystyle \frac{g^{'}_{\textnormal{num}}(\theta)}{g^{'}_{\textnormal{den}}(\theta)} \quad \textnormal{(2)}
\end{split}
\end{equation*}
\end{linenomath}
\noindent where,
\begin{linenomath}
\begin{equation*}
g^{'}_{\textnormal{den}}(\theta) = \left( \theta + \displaystyle \frac{b_{0}}{a_{0}} \right) \displaystyle \prod_{j \in \mathcal{U}} \left(\theta + \displaystyle \frac{b_{j}}{a_{j}} \right),
\end{equation*}
\end{linenomath}

\noindent and
\begin{linenomath}
\begin{equation*}
\begin{split}
& g^{'}_{\textnormal{num}}(\theta) \\ &= \left( \theta + \displaystyle \frac{b_{0}}{a_{0}} \right) \displaystyle \sum_{j \in \mathcal{U}} \displaystyle \prod_{\ell \in \mathcal{U} \backslash j} \left( \theta + \frac{b_{\ell}}{a_{\ell}}\right)  - (n - i + 1) \prod_{j \in \mathcal{U}} \left(\theta  + \displaystyle \frac{b_{j}}{a_{j}} \right) \;\textnormal{(3)}
\end{split}
\end{equation*}
\end{linenomath}
\noindent Since $a_{0} > 0, \;b_{0} \geq 0$, \; $b_{j} \geq 0$ for all $j \in \mathcal{N}^{\;'}$, and $f(\theta) > 0$ for $\theta \in [0,1]$, the denominator of Equation (2), that is,  $g^{'}_{\textnormal{den}}(\theta) > 0$. Hence, the sign of $g^{'}(\theta)$ is determined by the sign of $g^{'}_{\textnormal{num}}(\theta)$. Substituting $r_{0} = \displaystyle \frac{b_{0}}{a_{0}}$, and $r_{j} = \displaystyle \frac{b_{j}}{a_{j}}$ for $j \in \mathcal{U}$ into Equation (3) we get,
\begin{linenomath}
\begin{equation*}
\begin{split}
g^{'}_{\textnormal{num}}(\theta) = (\theta + r_{0}) \displaystyle \sum_{j \in \mathcal{U}} \displaystyle \prod_{\ell \in \mathcal{U} \backslash j} (\theta + r_{\ell}) - (n - i + 1)  \displaystyle \prod_{j \in \mathcal{U}} (\theta + r_{j}) \\\textnormal{(4)}.
\end{split}
\end{equation*}
\end{linenomath}

\noindent The function $g^{'}_{\textnormal{num}}(\theta)$ is a polynomial of degree $\overline{n}$. Let $e_{k}$ denote the \textit{k}-th elementary symmetric function in $r_{i}, r_{i+1}, ..., r_{n}$, then
\begin{linenomath}
\begin{equation*}
e_{k} \triangleq \displaystyle \sum_{1 \leq j_{1} < j_{2} < ...< j_{k} \leq n} r_{j_{1}} r_{j_{2}} \cdot \cdot \cdot r_{j_{k}}.    
\end{equation*}
\end{linenomath}

\noindent For example,
\begin{linenomath}
\begin{equation*}
\begin{split}
e_{0} & = 1 \\
e_{1} & = \displaystyle \sum_{i \leq j \leq n} r_{j} \\
& \vdots\\
e_{n-(i-1)} & = r_{i}r_{i+1} \cdots r_{n}\\
\end{split}    
\end{equation*}
\end{linenomath}

\noindent Now, we use the elementary symmetric functions $e_{k}$ in expressing the coefficients of the polynomial $g^{'}_{\textnormal{num}}(\theta)$. Expanding $g^{'}_{\textnormal{num}}(\theta)$ in Equation (4) gives us
\begin{linenomath}
\begin{equation*}
g^{'}_{\textnormal{num}}(\theta) = z_{0} \theta^{0} + z_{1} \theta + ... + z_{\overline{n}-1} \theta^{\overline{n}-1} + (-\zeta) \theta^{\overline{n}}    
\end{equation*}
\end{linenomath}

\noindent where for $0 \leq \; k \leq \; \overline{n} -1$, and
\begin{linenomath}
\begin{equation*}
z_{k} = (k+1)\; r_{0} \; e_{\overline{n}-(k+1)}  - (\overline{n} + \zeta -k) \; e_{\overline{n}-k} \quad \textnormal{(5)}    
\end{equation*}
\end{linenomath}

\noindent As $g^{'}_{\textnormal{num}}(\theta)$ is a continuous function, it can only change signs at points where $g^{'}_{\textnormal{num}}(\theta) = 0$. Hence, we need to show that if $g^{'}_{\textnormal{num}}(0) < 0$, then there are no roots for $g^{'}_{\textnormal{num}}(\theta) = 0$. On the other hand, if  $g^{'}_{\textnormal{num}}(0) > 0$ and changes sign from positive to negative at some point $p \geq 0$, then there are no roots other than $\theta = p$ for $g^{'}_{\textnormal{num}}(\theta) = 0$. To examine the number of roots of the polynomial $g^{'}_{\textnormal{num}}(\theta)$, we apply Descartes' Rule of Signs which is stated as follows. If the terms of a polynomial with real coefficients are ordered from lowest degree variable to the highest degree variable, then the number of positive real roots of the polynomial is equal to the number of sign changes between consecutive coefficients. Therefore, 
\begin{enumerate}
    \item If $g^{'}_{\textnormal{num}}(0) = z_{0} < 0$, then we have to show that the remaining coefficients of the polynomial are negative, that is, $z_{1}, z_{2}, ..., z_{\overline{n}-1} \leq 0$.  
    
    \item If $g^{'}_{\textnormal{num}}(0) > 0$, and it changes sign at a point $\theta = p$ meaning there is a sign change from positive to negative at some coefficient of $\theta^{k}$. Then, we need to show that the remaining coefficients of the polynomial are negative, that is,  \\$z_{k+1}, z_{k+2}, ..., z_{\overline{n}-1} \leq 0$. 
\end{enumerate}

\noindent We observe that the first condition can be proved by substituting $k = 0$ in the second condition. From Equation (5), the condition of a negative coefficient can be equivalently expressed as a bound $s_{k}$ on $r_{0}$, that is,\\
\begin{linenomath}
\begin{equation*}
z_{k} \leq 0 \iff r_{0} \leq \displaystyle \frac{(\overline{n} + \zeta - k) \; e_{\overline{n} - k}}{(k+1) \; e_{\overline{n}-(k+1)}} \; \triangleq \; s_{k}    
\end{equation*}
\end{linenomath}

\noindent Since, the condition of a negative coefficient can be expressed as a bound on $r_{0}$, the remaining coefficients of $g^{'}_{\textnormal{num}}(\theta)$ are negative if the bound $s_{k}$ is tighter than the bounds $s_{k+1}, s_{k+2}, ..., s_{\overline{n}-1}$ placed by the remaining coefficients. A sufficient condition for this is achieved when the sequence $s_{k}$ is monotone non-decreasing. We show this by proving $s_{k-1} \leq s_{k}$, which is expressed as\\
\begin{linenomath}
\begin{equation*}
\displaystyle \frac{(\overline{n} + \zeta -(k-1))\; e_{\overline{n} - (k-1)}}{k \cdot e_{\overline{n} - k}} \leq \displaystyle \frac{(\overline{n} + \zeta - k) \; e_{\overline{n} - k}}{(k+1) e_{\overline{n} - (k+1)}} 
\end{equation*}
\end{linenomath}

\noindent The above inequality simplifies to,
\begin{linenomath}
\begin{equation*}
(\overline{n} + \zeta - k + 1) (k+1) e_{\overline{n} - k + 1} e_{\overline{n} - k - 1} \leq k (\overline{n} + \zeta - k) (e_{\overline{n}-k})^{2}    
\end{equation*}
\end{linenomath}

\noindent Substituting $\overline{n} - k = y$ we get,
\begin{linenomath}
\begin{equation*}
\begin{split}
(y + \zeta + 1) (\overline{n}- y+1) e_{y+1} \cdot e_{y-1} \leq (\overline{n} - y) (y + \zeta) (e_{y})^{2} \; \textnormal{; or,}
\end{split}
\end{equation*}
\end{linenomath}

\begin{linenomath}
\begin{equation*}
\displaystyle \frac{(e_{y})^{2}}{e_{y+1} e_{y-1}} \geq  \displaystyle \frac{(\overline{n} - y +1) (y + \zeta + 1)}{(y + \zeta)(\overline{n} - y)} \quad \textnormal{(6)}    
\end{equation*}
\end{linenomath}

\noindent Newton's inequality on elementary symmetric functions $e_{y}$ states
\begin{linenomath}
\begin{equation*}
\displaystyle \frac{(e_{y})^{2}}{e_{y+1} e_{y-1}} \geq  \displaystyle \frac{(\overline{n} - y +1) (y + 1)}{ y (\overline{n} - y)}.
\end{equation*}
\end{linenomath}

\noindent Newton's inequality is sufficient to show that Equation (6) is satisfied if,
\begin{linenomath}
\begin{equation*}
\begin{split}
\displaystyle \frac{(\overline{n} - y +1) (y + 1)}{ y (\overline{n} - y)} & \geq \displaystyle \frac{(\overline{n} - y +1) (y + \zeta + 1)}{(y + \zeta)(\overline{n} - y)} \; \textnormal{; or} \\[4 pt]
(y + 1) (y + \zeta) \; & \geq \; y (y + \zeta + 1) \quad \textnormal{; or}\\
\zeta & \geq 0\
\end{split}    
\end{equation*}
\end{linenomath}

\end{proof}

As mentioned in the main manuscript, in addition to resource constraints, the LCP optimization is based on optimizing on one time interval while holding the others constant. Hence, we also have constraints to ensure that all agents finish in the correct order. 

\begin{customexample}{1} \label{lcp-tie}

There exist LCP allocations where in some time intervals the allocations of agents are not the extreme points of the Pareto frontier formed by the set of participating agents in that interval. The allocations of agents in these intervals are the extreme points of the feasible set where the constraints of completion times are tight, which represent ties. 

\noindent Consider the demand matrix $\mathbf{D}$, tasks matrix $\mathbf{K}$ for a set of agents $\mathcal{N} = \{1,2, ...,n\}$ as follows,\\

$\mathbf{D} = \begin{bmatrix} 1 \quad & \dfrac{2}{3} \quad & \epsilon 
\\[6 pt] \epsilon \quad & \dfrac{1}{2} \quad & 1
\\[6 pt] 1    \quad & 1 \quad & 1 
\\ 1    \quad & 1 \quad & 1
\\ \vdots \quad & \vdots \quad & \vdots
\\ 1    \quad & 1 \quad & 1\end{bmatrix}_{n \times 3}, \quad$\\[6 pt]

$\mathbf{K} = \begin{bmatrix} & 1 \quad & 0 \quad & 0 \quad & 0 \quad \cdots \quad & 0 
\\ & 0 \quad & 1 \quad & 0 \quad & 0 \quad \cdots \quad & 0
\\ & 0 \quad & 0 \quad & 5 \quad & 0 \quad \cdots \quad & 0 
\\ & 0 \quad & 0 \quad & 0 \quad & 5 \quad \cdots \quad & 0
\\ & \vdots \quad & \vdots \quad & \vdots \quad & \vdots \quad & \vdots
\\& 0 \quad & 0 \quad  & 0 \quad & 0 \quad \cdots \quad & 5 
\end{bmatrix}_{n \times n}$\\[6 pt]

\noindent where $0 < \epsilon << 1$ is a very small non-zero positive constant. For this example, we consider $\epsilon = 0$ since $\epsilon > 0$ alters the solution by a negligible amount. 

Agents $i \geq 3$ have a dominant demand for each of the three resources and a large amount of work to complete compared to agents $ i = 1$ and $ i = 2$. Hence, agents $i \geq 3$ will be allocated resources serially after agents $i = 1$ and  $i = 2$, and the completion times of agents $ i = 1$ and $ i = 2$ must be as small as possible in order to minimize the cost product. We show that when $n$ is sufficiently large the minimum cost product is obtained by minimizing the makespan of agents $ i = 1$ and $ i = 2$. 

Let $\mathbf{A}^{\textnormal{tie}}(\cdot)$ denote the allocation that minimizes the makespan for the agents $i = 1$ and  $i = 2$. The allocation $\mathbf{A}^{\textnormal{tie}}(\cdot)$ is obtained with $\lambda_{11} = \lambda_{21} = \dfrac{6}{7} \;$ causing both the agents to have the same completion time of $t_{1} = t_{2} = \dfrac{7}{6}$. The allocation $\mathbf{A}^{\textnormal{tie}}(\cdot)$ during the first interval is as follows.
\begin{linenomath}
\begin{equation*}
\begin{split}
\mathbf{A}^{\textnormal{tie}}_{1}(t) = \langle \dfrac{6}{7}, \; \dfrac{4}{7}, \;  0 \rangle, \quad \forall t \in [0, \; \dfrac{7}{6}) \\[4 pt]
\mathbf{A}^{\textnormal{tie}}_{2}(t) = \langle 0, \; \dfrac{3}{7}, \;  \dfrac{6}{7} \rangle, \quad \forall t \in [0, \; \dfrac{7}{6})
\end{split}    
\end{equation*}
\end{linenomath}

\noindent In the allocation $\mathbf{A}^{\textnormal{tie}}(\cdot)$ only the second resource gets saturated in the first time interval.

There exists two allocations that have a lower cost product when considering only agents $i = 1$ and $i = 2$ that correspond to the allocations obtained from the extreme points of the Pareto frontier. However, the completion time of the second agent to complete its work in the two allocations is greater than the completion of the second agent in $\mathbf{A}^{\textnormal{tie}}(\cdot)$. These two allocations are,
\begin{itemize}
    \item Let $\mathbf{A}^{\textnormal{extreme-1}}(\cdot)$ denote the allocation obtained from saturating the first and the second resource with $\lambda_{11} = 1$ and $\lambda_{21} = \dfrac{2}{3}$. The completion times of the agents in $\mathbf{A}^{\textnormal{extreme-1}}(\cdot)$ are $t_{1} = 1$ and $t_{2} = \dfrac{4}{3}$.\\
    \begin{linenomath}
    \begin{equation*}
    \begin{split}
    \mathbf{A}^{\textnormal{extreme-1}}_{1}(t) & = \langle 1, \; \dfrac{2}{3}, \;  0 \rangle, \quad \forall t \in [0, \; 1)\\[4 pt] 
    \mathbf{A}^{\textnormal{extreme-1}}_{2}(t) & = \begin{cases} \langle 0, \; \dfrac{1}{3}, \; \dfrac{2}{3} \rangle, \quad \forall t \in [0, \; 1)\\[8 pt] \langle 0, \; \dfrac{1}{2}, \; 1\rangle, \quad \forall t \in [1, \; \dfrac{4}{3})\\[4 pt] \end{cases}\\   
    \end{split}    
    \end{equation*}
    \end{linenomath}
    
    \item Let $\mathbf{A}^{\textnormal{extreme-2}}(\cdot)$ denote the allocation obtained from saturating the second and the third resource with $\lambda_{11} = \dfrac{3}{4}$ and $\lambda_{21} = 1$. The completion times of the agents in $\mathbf{A}^{\textnormal{extreme-2}}(\cdot)$ are $t_{2} = 1$ and $t_{1} = \dfrac{5}{4}$.\\
    \begin{linenomath}
    \begin{equation*}
    \begin{split}
    \mathbf{A}^{\textnormal{extreme-2}}_{1}(t) & = \begin{cases} \langle \dfrac{3}{4}, \; \dfrac{1}{2}, \; 0 \rangle, \quad \forall t \in [0, \; 1)\\[8 pt] 
    \langle 1, \; \dfrac{2}{3}, \; 0\rangle, \quad \forall t \in [1, \; \dfrac{5}{4}) \\[8 pt] \end{cases} \\[8 pt]
    \mathbf{A}^{\textnormal{extreme-2}}_{2}(t) & = \langle 0, \; \dfrac{1}{2}, \;  1 \rangle, \quad \forall t \in [0, \; 1)\\[4 pt]
    \end{split}    
    \end{equation*}
    \end{linenomath}
    
\end{itemize}

The allocations of agents $i \geq 3$ in the LCP allocation will be in a serial order after the work of agents $i = 1$ and $i = 2$ have been completed. We show that for a sufficiently large $n$ the allocation $\mathbf{A}^{\textnormal{tie}}(\cdot)$ results in a lower log cost product when compared to the allocations formed by the extreme points of the Pareto frontier, that is, $\mathbf{A}^{\textnormal{extreme-1}}(\cdot)$ and $\mathbf{A}^{\textnormal{extreme-2}}(\cdot)$. The log cost product of the allocations $\mathbf{A}^{\textnormal{tie}}(\cdot)\;$, $\mathbf{A}^{\textnormal{extreme-1}}(\cdot)\;$, and $\mathbf{A}^{\textnormal{extreme-2}}(\cdot)\;$ are listed as follows.
\begin{linenomath}
\begin{equation*}
\log(\mathbf{A}^{\textnormal{tie}}(\cdot)) = \log(\dfrac{7}{6}) + \log(\dfrac{7}{6}) + \displaystyle \sum_{i = 3}^{n} \log(5 * (i -2) + \dfrac{7}{6})
\end{equation*}
\end{linenomath}
\begin{linenomath}
\begin{equation*}
    \log(\mathbf{A}^{\textnormal{extreme-1}}(\cdot)) = \log(1) + \log(\dfrac{4}{3}) + \displaystyle \sum_{i = 3}^{n} \log(5 * (i -2) + \dfrac{4}{3}).
\end{equation*}
\end{linenomath}
\begin{linenomath}
\begin{equation*}
    \log(\mathbf{A}^{\textnormal{extreme-2}}(\cdot)) = \log(1) + \log(\dfrac{5}{4}) + \displaystyle \sum_{i = 3}^{n} \log(5 * (i -2) + \dfrac{5}{4}) \\[4 pt]
\end{equation*}
\end{linenomath}

Let $\mathbf{A}^{\textnormal{extreme}}(\cdot)$ denote both $\mathbf{A}^{\textnormal{extreme-1}}(\cdot)$ and $\mathbf{A}^{\textnormal{extreme-2}}(\cdot)$. Since, 
\begin{linenomath}
\begin{equation*}
    \displaystyle \dfrac{\mathbf{A}^{\textnormal{tie}}(\cdot)}{\mathbf{A}^{\textnormal{extreme}}(\cdot)}  <  1  \iff \log(\mathbf{A}^{\textnormal{tie}}(\cdot)) - \log(\mathbf{A}^{\textnormal{extreme}}(\cdot)) < 0
\end{equation*}
\end{linenomath}

\noindent we programmed the values of the completion times of the allocations $\mathbf{A}^{\textnormal{tie}}(\cdot), \;$ $\mathbf{A}^{\textnormal{extreme-1}}(\cdot),\;$ and $\mathbf{A}^{\textnormal{extreme-2}}(\cdot)\;$ for increasing values of $n$. We found that 
\begin{linenomath}
\begin{equation*}
\begin{split}
\log(\mathbf{A}^{\textnormal{tie}}(\cdot)) - \log(\mathbf{A}^{\textnormal{extreme-1}}(\cdot)) & < 0, \quad \forall n \geq 3 \\[2 pt]
\log(\mathbf{A}^{\textnormal{tie}}(\cdot)) - \log(\mathbf{A}^{\textnormal{extreme-2}}(\cdot)) & < 0, \quad \forall n \geq 132\\[8 pt]
\end{split}
\end{equation*}
\end{linenomath}

\end{customexample}

\begin{customlemma}{4}
The LCP mechanism violates SP.
\end{customlemma}

\begin{proof}
We provide an example illustrating that the LCP mechanism is not Strategy-Proof. We show that an agent is able to reduce its cost by misreporting its demand. Consider a set of two agents, $\mathcal{N} = \{1, 2\}$.\\

\noindent \textit{Costs when Agents Report True Demands}:\\
\noindent Demands, tasks, and amount of work required to be completed are as follows:\\
$\mathbf{D} = \begin{bmatrix} \dfrac{1}{2} \quad 1 \\1 \quad \dfrac{1}{6} \end{bmatrix}, \quad$
$\mathbf{K} = \begin{bmatrix} 1 \quad 0 \\0 \quad 1 \end{bmatrix}, \quad \textnormal{and}$
$\mathbf{W} = \begin{bmatrix} \dfrac{1}{2} \quad 1 \\1 \quad \dfrac{1}{6} \end{bmatrix}$.\\[2 pt]

\noindent The output allocation from the LCP mechanism can be obtained by solving the optimization problem,
\begin{linenomath}
\begin{align*}
\min  \quad  &  \dfrac{1}{\lambda_{1,1}}\cdot \big( 1 - \dfrac{\lambda_{2,1}}{\lambda_{1,1}} + \dfrac{1}{\lambda_{1,1}}  \big) 
\\ \textnormal{subject to} \quad  & \dfrac{\lambda_{1,1}}{2} + \lambda_{2,1} \leq 1, 
\\\quad & \lambda_{1,1} + \dfrac{\lambda_{2,1}}{6} \leq 1,
\end{align*}
\end{linenomath}

\noindent where $\lambda_{1,1},\; \lambda_{2,1}$ are the dominant shares of agent 1 and 2 respectively. The solution is $\lambda_{1,1} = \dfrac{10}{11}$, and $\lambda_{2,1} = \dfrac{6}{11}$ till agent 1 completes its work. Later, agent 2 would be allocated all the resources till it finishes the remaining amount of work, that is, $\lambda_{2,2} = 1$. The allocations of the two agents are,
\begin{linenomath}
\begin{equation*}
\begin{split}
\mathbf{A}^{\textnormal{LCP}}_{1}(t) & = \langle \dfrac{5}{11}, \; \dfrac{10}{11} \rangle, \quad \forall t \in [0, \; \dfrac{11}{10}) \\[4 pt]    
\mathbf{A}^{\textnormal{LCP}}_{2}(t) & = \begin{cases}
\langle \dfrac{6}{11}, \; \dfrac{1}{11} \rangle, \quad \forall t \in [0, \; \dfrac{11}{10})\\[8 pt]
\langle 1, \; \dfrac{1}{6} \rangle, \quad \forall t \in [\dfrac{11}{10}, \; \dfrac{15}{10}] \\[8 pt]
\end{cases}
\end{split}    
\end{equation*}
\end{linenomath}

\noindent The above allocation results in the following costs, \\ $c_{1}(\mathbf{A}^{\textnormal{LCP}}_{1}(\cdot)) = \dfrac{11}{10}$, and $c_{2}(\mathbf{A}^{\textnormal{LCP}}_{2}(\cdot)) = 1 - \dfrac{6}{10} + \dfrac{11}{10} = \dfrac{15}{10}$. \\[12 pt]

\noindent \textit{Costs when an Agent Misreports its Demands}:\\ \noindent Now, agent 1 changes it's demand to $\langle \dfrac{2}{3},\; 1 \rangle$ resulting in the following misreported demand matrix, and amount of work required to be completed, \\
$\mathbf{D}^{'} = \begin{bmatrix} \dfrac{2}{3} \quad 1 \\1 \quad \dfrac{1}{6} \end{bmatrix}, \quad \textnormal{and}$
$\mathbf{W}^{'} = \begin{bmatrix} \dfrac{2}{3} \quad 1 \\1 \quad \dfrac{1}{6} \end{bmatrix}$.\\

The objective of the optimization stays the same as in the case when original demands were reported, but the constraints change. 
\begin{linenomath}
\begin{align*}
\min \quad & \dfrac{1}{\lambda_{1,1}}\cdot \big( 1 - \dfrac{\lambda_{2,1}}{\lambda_{1,1}} + \dfrac{1}{\lambda_{1,1}}  \big) 
\\ \textnormal{subject to} \quad & \dfrac{\lambda_{2,1}}{3} + \lambda_{2,1} \leq 1, 
\\\quad & \lambda_{1,1} + \dfrac{\lambda_{2,1}}{6} \leq 1,
\end{align*}.
\end{linenomath}

\noindent The solution is, $\lambda^{'}_{1,1} = \dfrac{15}{16}$, and $\lambda^{'}_{2,1} = \dfrac{3}{8}$ till agent 1 completes its work.  Then, $\lambda^{'}_{2,2} = 1$ till agent 2 completes its remaining work. The allocation due to misreported demand is depicted as follows, \\
\begin{linenomath}
\begin{equation*}
\begin{split}
\mathbf{A}^{*\textnormal{LCP}}_{1}(t) & = \langle \dfrac{5}{8}, \; \dfrac{15}{16} \rangle, \quad \forall t \in [0, \; \dfrac{16}{15}) \\[6 pt]
\mathbf{A}^{*\textnormal{LCP}}_{2}(t) & = \begin{cases}
\langle \dfrac{3}{8}, \; \dfrac{1}{16} \rangle, \quad \forall t \in [0, \; \dfrac{16}{15}) \\[8 pt]
\langle 1, \; \dfrac{1}{6} \rangle, \quad \forall t \in [\dfrac{16}{15}, \; \dfrac{5}{3}] \\
\end{cases}\\[6 pt] 
\end{split}    
\end{equation*}
\end{linenomath}

\noindent The allocations  due to the misreported demands have the following costs,\\ $c_{1}(\mathbf{A}^{*\textnormal{LCP}}_{1}(\cdot)) = \dfrac{16}{15}$,  and $c_{2}(\mathbf{A}^{*\textnormal{LCP}}_{2}(\cdot)) = 1 - \dfrac{2}{5} + \dfrac{16}{15} = \dfrac{5}{3}$. \\[6 pt]

\noindent We observe that agent 1 is able to decrease its cost by misreporting its demands, $c_{1}(\mathbf{A}^{*\textnormal{LCP}}_{1}(\cdot)) < c_{1}(\mathbf{A}^{\textnormal{LCP}}_{1}(\cdot))$.

\end{proof}

\subsection{Envy-freeness of LCP}

We will first present the LCP allocation which is not envy-free, even in expectation.

\begin{customlemma}{5}
The LCP mechanism does not satisfy EF (even in expectation) when the resource allocation involves three or more agents and two or more resources. 
\end{customlemma}

\begin{proof}
We provide an example an LCP allocation among three agents that exhibits envy. Consider the following demand matrix $\mathbf{D}$, and tasks matrix $\mathbf{K}$ for a set of three agents $\mathcal{N} = \{1, 2, 3\}$, \\

\noindent $\mathbf{D} = \begin{bmatrix} 1 \quad & 1 \\
1 \quad & \epsilon \\
\epsilon \quad & 1\end{bmatrix} ,\quad$\\[2 pt]

\noindent $\mathbf{K} = \begin{bmatrix} 1 \quad & 0 \quad & 0 \\
0 \quad & 1 \quad & 0 \\
0 \quad & 0 \quad & k_{3} \\\end{bmatrix}, \quad$\\[2 pt]

\noindent where $\epsilon << 1$ and $k_{3} > 2$.

\noindent We apply Theorem \ref{cost-prodLCP} and examine the possible LCP allocations.

\begin{itemize}
    \item The Shortest Job First allocation $\mathbf{A}^{\textnormal{SJF}}(\cdot)$. The agents are allocated resources in the sorted order of their tasks. The completion times of the agents are $t_{1} = 1,\; t_{2} = 1 + 1  = 2$, and $t_{3} = 2 + k_{3}$. The cost product of the allocation $\mathbf{A}^{\textnormal{SJF}}(\cdot)$ is
    \begin{linenomath}
    \begin{equation*}
    CP(\mathbf{A}^{\textnormal{SJF}}(\cdot)) = 1 \cdot 2 \cdot (2 + k_{3}) = 4 + 2 k_{3}.
    \end{equation*}
    \end{linenomath}
    
    \item The allocation $\mathbf{A}^{\textnormal{s-1}}(\cdot)$ where agents 2 and 3 share in the first interval, agent 1 is allocated in the second time interval, and remaining work of agent 3 is completed in the third interval. The completion times of the agents are $t_{2} = 1 + \epsilon, \; t_{1} = 1 + 1 + \epsilon = 2 + \epsilon, \;$ and $t_{3} = 2 + \epsilon + k_{3} - 1$. The cost product of the allocation $\mathbf{A}^{\textnormal{s-1}}(\cdot)$ is
    \begin{linenomath}
    \begin{equation*}
    CP(\mathbf{A}^{\textnormal{s-1}}(\cdot)) = (1+ \epsilon) \cdot (2 + \epsilon) \cdot (k_{3} + \epsilon + 1).
    \end{equation*}
    \end{linenomath}
    
    \item The allocation $\mathbf{A}^{\textnormal{s-2}}(\cdot)$ where agents 2 and 3 share in the first interval, agent 3 completes its remaining work in the second time interval, and agent 1 is allocated in the third interval. The completion times of the agents are $t_{2} = 1 + \epsilon, \; t_{3} = 1 + \epsilon + k_{3} - 1 = k_{3} + \epsilon, \;$ and $t_{1} = 1 + \epsilon + k_{3}$. The cost product of the allocation $\mathbf{A}^{\textnormal{s-2}}(\cdot)$ is
    \begin{linenomath}
    \begin{equation*}
    CP(\mathbf{A}^{\textnormal{s-2}}(\cdot)) = (1+ \epsilon) \cdot (k_{3} + \epsilon) \cdot (k_{3} + \epsilon + 1).     
    \end{equation*}
    \end{linenomath}
    
    \item The allocation $\mathbf{A}^{\textnormal{s-3}}(\cdot)$ where agent 1 is allocated in the first interval, agents 2 and 3 share in the second time interval, and agent 3 completes its remaining work in the third interval. The completion times of the agents are $t_{1} = 1, \; t_{2} =  1 + 1 + \epsilon = 2 + \epsilon, \;$ and $t_{3} = 2 + \epsilon + k_{3} - 1 = 1 + \epsilon + k_{3}$. The cost product of the allocation $\mathbf{A}^{\textnormal{s-3}}(\cdot)$ is
    \begin{linenomath}
    \begin{equation*}
    CP(\mathbf{A}^{\textnormal{s-3}}(\cdot)) = (1) \cdot (2 + \epsilon) \cdot (1 + \epsilon + k_{3}).     
    \end{equation*}
    \end{linenomath}
    
\end{itemize}

\noindent If $\epsilon << 1$ and $k_{3} > 2 \;$ such that
$k_{3} < \dfrac{2}{\epsilon} - 3 - \epsilon \;$
then, \\
\begin{linenomath}
\begin{equation*}
    CP(\mathbf{A}^{\textnormal{s-3}}(\cdot)) < CP(\mathbf{A}^{\textnormal{SJF}}(\cdot)) <  CP(\mathbf{A}^{\textnormal{s-1}}(\cdot)) < CP(\mathbf{A}^{\textnormal{s-2}}(\cdot)).\\
\end{equation*}
\end{linenomath}

\noindent The LCP allocation is $\mathbf{A}^{\textnormal{s-3}}(\cdot)$ where agent 2 envies agent 1 because $c_{2}(\mathbf{A}^{\textnormal{LCP}}_{1}(\cdot)) = 1 <  c_{2}(\mathbf{A}^{\textnormal{LCP}}_{2}(\cdot)) = 2 + \epsilon$.

\end{proof}

We study the LCP mechanism in detail for the two cases where the LCP allocations satisfy envy-freeness in expectation as per Lemma \ref{lcp-1res_and_2agents}.

\subsection{LCP Mechanism for Single Resource Allocations} \label{single-resource-lcp}

In a single resource allocation problem, all \textit{n} agents have the same demand but different amounts of work to be completed. The model reduces to  an offline allocation problem similar to the work of \citeauthor{friedman2003fairness} \shortcite{friedman2003fairness} who study the online model of allocating a single resource, where an agent arrives dynamically with their required service time. We present the results derived from applying the LCP mechanism on single resource allocations.

\begin{customlemma}{A.1} \label{sjf-lcp}
For single resource allocations, the LCP mechanism allocates the resource completely to the agents in a serial order based on the increasing amounts of work required to be completed. In other words, the allocations are of the form of the Shortest Job First (SJF).
\end{customlemma}

\begin{proof}
We first show that the LCP allocation is a form of serial allocation where a single resource is allocated completely to the agents in a serial order. Next, we show that the serial order must be a sorted increasing order based on the amount of work required to be completed by each agent.

Let us assume for the sake of contradiction that the LCP allocation $\mathbf{A}^{\textnormal{LCP}}(\cdot)$ is a form of shared allocation where one or more agents share the resource with other agents in any manner. Without loss of generality, let the completion times of the agents in the shared LCP allocation be in the order $t_{1} \leq t_{2} \leq ... \leq t_{n}$. That is, $t_{1}$ denotes the completion time of the agent that completes first under the shared LCP allocation, $t_{2}$ denotes the completion time of the agent that completes second, and so on.

We now transform this allocation into another allocation where there is no sharing among the agents. This transformation is done by allocating all of the shared resource completely to an agent based on the order of the completion times in the shared LCP allocation. This means, if $t_{i} < t_{j}$ in the shared LCP allocation, the resource is completely allocated to agent \textit{i} before allocating them to agent \textit{j}. Let the completion time of an agent \textit{i} in the transformed allocation be represented as $t^{'}_{i}$. In order to show that this transformed allocation achieves a lower cost product we need to prove,
\\$t^{'}_{i} \leq t_{i}, \quad \forall i \in \mathcal{N}$.

If $\mathbf{w}_{k}$ denotes the amount of work that is required to be completed by an agent \textit{k}, we know that,
\begin{linenomath}
\begin{equation*}
\begin{split}
\mathbf{w}_{k} = \displaystyle \int^{t_{k}}_{0} \mathbf{A}_{k}(t)dt \quad \textnormal{or,}\\
\mathbf{w}_{k} = \displaystyle \sum^{k}_{j = 1} \Delta t_{j} A^{\Delta t_{j}}_{k} \quad,
\end{split}    
\end{equation*}
\end{linenomath}

\noindent where $A^{\Delta t_{j}}_{k}$ represents the resource allocated by the shared LCP mechanism to the agent \textit{k} for the interval $[t_{j-1}, \; t_{j})$. Summing the amount of work completed by the first \textit{i} agents we get,\\
\begin{linenomath}
\begin{equation*}
\begin{split}
\displaystyle \sum^{i}_{k=1} \mathbf{w}_{k} & = \sum^{i}_{k=1} \; \sum^{k}_{j=1} \Delta t_{j} A^{\Delta t_{j}}_{k} \quad \textnormal{or,}\\
\displaystyle \sum^{i}_{k=1} \mathbf{w}_{k}  & \leq \; \sum^{i}_{k=1} \; \sum^{i}_{j=1} \Delta t_{j} A^{\Delta t_{j}}_{k} \quad \textnormal{or,}\\
\displaystyle \sum^{i}_{k=1} \mathbf{w}_{k} & \leq \;  \sum^{i}_{j=1} \Delta t_{j} \; \sum^{i}_{k=1} A^{\Delta t_{j}}_{k} \quad \textnormal{or, }\\
\displaystyle \sum^{i}_{k=1} \mathbf{w}_{k} & \leq \;  \sum^{i}_{j=1} \Delta t_{j} \quad \textnormal{or,}\\
\displaystyle \sum^{i}_{k=1} \mathbf{w}_{k} & \leq  \; t_{i}
\end{split}    
\end{equation*}
\end{linenomath}

\noindent since $\displaystyle \sum^{i}_{k=1} A^{\Delta t_{j}}_{k} \; \leq \; 1$ over each $\Delta t_{j}$ and $\displaystyle \sum^{i}_{j=1} \Delta t_{j} = t_{i} $. 

Hence, when only a single resource is allocated the lower bound on the completion time of an agent \textit{i} is obtained as $t_{i} \; \geq \; \sum^{i}_{k=1}\mathbf{w}_{k}$. We observe that under the transformation described above, the resource is completely allocated to the agents in a serial order, that is, $A^{\Delta t_{i}}_{i} = 1$ for agent \textit{i} and $A^{\Delta t_{j}}_{i} = 0$ for $j \neq i$. This achieves the lower bound for every agent because $t^{'}_{i} = \sum^{i}_{k=1}\mathbf{w}_{k} \leq t_{i}, \quad \forall i \in \mathcal{N}$. Hence, the LCP allocation is in the form of a serial resource allocation.  

Now, we move on to showing that the serial order must be in the form of Shortest Job First. For the sake of contradiction, let us assume that the LCP serial allocation is denoted by the sequence $i^{\textnormal{lcp}}$, which is an arbitrary sequence with no regard to the amount of work required to be completed by an agent. Then, the agent completion times will be in an order denoted as follows, $t_{i^{\textnormal{lcp}}_{1}} \leq t_{i^{\textnormal{lcp}}_{2}} \leq ... \leq t_{i^{\textnormal{lcp}}_{n}}$. Since the sequence $i^{\textnormal{lcp}}$ is not based on the amount of work required to be completed by an agent, there exists a position \textit{p} in the sequence such that, $\mathbf{w}_{i^{\textnormal{lcp}}_{p-1}} \geq \mathbf{w}_{i^{\textnormal{lcp}}_{p}}$. We can construct a new serial allocation denoted by the sequence $i^{\textnormal{s}}$ as follows. The sequence $i^{\textnormal{s}}$ is almost the same as $i^{\textnormal{lcp}}$ except the agents at positions $p-1$ and $p$ are swapped. Mathematically,
\begin{linenomath}
\begin{equation*}
\begin{split}
i^{\textnormal{s}}_{j} = \begin{cases}
i^{\textnormal{lcp}}_{j}, \quad j = 1,2, ..., p-2 \\[4 pt]
i^{\textnormal{lcp}}_{p}, \quad j = p-1 \\[4 pt]
i^{\textnormal{lcp}}_{p-1}, \quad j = p \\[4 pt]
i^{\textnormal{lcp}}_{j}, \quad j = p+1, ..., n \\
\end{cases}\\[4 pt]    
\end{split}    
\end{equation*}
\end{linenomath}

\noindent From the previous portion of the proof, we know that the completion time of an agent in the assumed LCP serial allocation is given by
\begin{linenomath}
\begin{equation*}
    t_{i^{\textnormal{lcp}}_{j}} =\displaystyle \sum^{i^{\textnormal{lcp}}_{j}}_{k=i^{\textnormal{lcp}}_{1}} \mathbf{w}_{k}
\end{equation*}
\end{linenomath}

\noindent The completion time of the agents under the newly constructed serial allocation $i^{\textnormal{s}}$ till position $p-2$ is the same as the LCP allocation. That is, $t_{i^{\textnormal{s}}_{j}} = t_{i^{\textnormal{lcp}}_{j}}, \quad j = 1, ..., p-2$. But at position $p-1$ we have,
\begin{linenomath}
\begin{equation*}
\begin{split}
t_{i^{\textnormal{s}}_{p-1}} & = \displaystyle \sum^{i^{\textnormal{s}}_{p-1}}_{k=i^{\textnormal{s}}_{1}} \mathbf{w}_{k} \quad  \textnormal{or,}\\
t_{i^{\textnormal{s}}_{p-1}} & = \displaystyle \sum^{i^{\textnormal{s}}_{p-2}}_{k=i^{\textnormal{s}}_{1}} \mathbf{w}_{k} + \mathbf{w}_{i^{\textnormal{s}}_{p-1}} \quad \textnormal{or,}\\
t_{i^{\textnormal{s}}_{p-1}} & = \displaystyle \sum^{i^{\textnormal{lcp}}_{p-2}}_{k=i^{\textnormal{lcp}}_{1}} \mathbf{w}_{k} + \mathbf{w}_{i^{\textnormal{lcp}}_{p}} \quad \textnormal{or,}\\
t_{i^{\textnormal{s}}_{p-1}} & \leq \displaystyle \sum^{i^{\textnormal{lcp}}_{p-2}}_{k=i^{\textnormal{lcp}}_{1}} \mathbf{w}_{k} + \mathbf{w}_{i^{\textnormal{lcp}}_{p-1}} \quad \textnormal{or,}\\
t_{i^{\textnormal{s}}_{p-1}} & \leq \displaystyle \sum^{i^{\textnormal{lcp}}_{p-1}}_{k=i^{\textnormal{lcp}}_{1}} \mathbf{w}_{k} \quad \textnormal{or,}\\
t_{i^{\textnormal{s}}_{p-1}} & \leq t_{i^{\textnormal{lcp}}_{p-1}}
\end{split}    
\end{equation*}
\end{linenomath}

\noindent At position \textit{p} we have,
\begin{linenomath}
\begin{equation*}
\begin{split}
t_{i^{\textnormal{s}}_{p}} & =\displaystyle \sum^{i^{\textnormal{s}}_{p}}_{k=i^{\textnormal{s}}_{1}} \mathbf{w}_{k} \quad \textnormal{or,}\\
t_{i^{\textnormal{s}}_{p}} & = \displaystyle \sum^{i^{\textnormal{s}}_{p-2}}_{k=i^{\textnormal{s}}_{1}} \mathbf{w}_{k} + \mathbf{w}_{i^{\textnormal{s}}_{p-1}} + \mathbf{w}_{i^{\textnormal{s}}_{p}} \quad \textnormal{or,}\\ t_{i^{\textnormal{s}}_{p}} & = \displaystyle \sum^{i^{\textnormal{lcp}}_{p-2}}_{k=i^{\textnormal{lcp}}_{1}} \mathbf{w}_{k} + \mathbf{w}_{i^{\textnormal{lcp}}_{p}} + \mathbf{w}_{i^{\textnormal{lcp}}_{p-1}} \quad \textnormal{or,}\\
t_{i^{\textnormal{s}}_{p}} & = \displaystyle \sum^{i^{\textnormal{lcp}}_{p}}_{k=i^{\textnormal{lcp}}_{1}} \mathbf{w}_{k} \quad \textnormal{or,}\\[2 pt]
t_{i^{\textnormal{s}}_{p}}  & = \; t_{i^{\textnormal{lcp}}_{p}}\quad \textnormal{so;}\\[2 pt]
t_{i^{\textnormal{s}}_{p}} & = t_{i^{\textnormal{lcp}}_{p}}\\
\end{split}    
\end{equation*}
\end{linenomath}

\noindent Similar to before we can show that, $t_{i^{\textnormal{s}}_{j}} = t_{i^{\textnormal{lcp}}_{j}}, \quad \textnormal{for}\; j = p+2, ..., n$.

The newly constructed serial allocation formed by swapping based on the sorted increasing order of the amount of work required to be completed by an agent has improved the completion time of the agent at position $p-1$ and kept the completion times of the remaining agents the same. Hence, we have reached a contradiction that the assumed serial allocation is the LCP serial allocation.
\end{proof}

\begin{customlemma}{A.2} \label{sjf-ef}
SJF allocation satisfies ex-ante Envy-freeness.
\end{customlemma}

\begin{proof}
Let us assume that the amount of work required to be completed by \textit{n} set of agents are in the order $k_{1} \leq k_{2} \leq ... \leq k_{n}$. Then, the SJF allocates the resource to agents in that order breaking ties randomly. The cost of the agents are given by 
\begin{linenomath}
\begin{align*}
\begin{split}
c_{1}(\mathbf{A}_{1}(\cdot)) & = k_{1}, \\
c_{2}(\mathbf{A}_{2}(\cdot)) & = k_{1} + k_{2}, \\
& \vdots \qquad \qquad \qquad \textnormal{(7)}\\
c_{n}(\mathbf{A}_{n}(\cdot)) & = \displaystyle \sum_{i = 1}^{n} k_{i} \\
\end{split}    
\end{align*}
 \end{linenomath}

\noindent In a slight abuse of notation, we refer to $w_{i}(\mathbf{A}_{j}(\cdot)) = \displaystyle \int_{0}^{t_{j}} \mathbf{A}_{j}(t) dt$ as the amount of work done by agent \textit{i} with agent \textit{j}'s allocation. An agent \textit{i} envies agent \textit{j} in allocation $\mathbf{A}(\cdot)$, if agent \textit{i} is able to complete its required amount of work $k_{i}$ at a (strictly) lower cost with agent \textit{j}'s allocation. Mathematically, agent \textit{i} envies agent \textit{j} if, $w_{i}(\mathbf{A}_{j}(\cdot)) \; \geq \; k_{i}$, and $ c_{i}(\mathbf{A}_{i}(\cdot)) \; > \; c_{i}(\mathbf{A}_{j}(\cdot))$. In other words, agent \textit{i} would have completed its required amount of work $k_{i}$ earlier with allocation $\mathbf{A}_{j}(\cdot)$. 

Now, we show that the SJF allocation is Envy-free in expectation. We know from Equation (7), that the completion time of an agent is ordered based on the amount of work it is required to complete. The only possibility of envy occurs when an agent \textit{i} completes its work later than agent \textit{j}, that is, $t_{i}\; > \;t_{j}$. Since agent \textit{i} is scheduled later, 
$w_{i}(\mathbf{A}_{j}(\cdot))  =  k_{j} \; \leq \; k_{i}$.
Hence, the envy occurs only when the agents have equal number of tasks. Since the agent that gets scheduled earlier when there are equal number of tasks is selected randomly, the SJF allocation is  Envy-free in expectation.
\end{proof}

\begin{customcorollary}{A.3} \label{single-res-ef}
For single resource allocations, the LCP mechanism is Envy-free in expectation.
\end{customcorollary}

\begin{proof}
The proof of this corollary follows directly from Lemma \ref{sjf-lcp} and Lemma \ref{sjf-ef}.
\end{proof}

\subsection{LCP Mechanism with Two Agents} \label{lcp-2-agents}

We now look at the the LCP mechanism for the allocation of multiple resources when only two agents are present. The cost product for allocation involving two agents is expressed as\\  
\begin{linenomath}
\begin{equation*}
    CP(\lambda_{1,1}, \lambda_{2,1}, \lambda_{2,2}) = \dfrac{k_{1}}{\lambda_{1,1}} \cdot \left(\displaystyle  \frac{k_{2} - \displaystyle \frac{k_{1}}{\lambda_{1,1}} \cdot \lambda_{2,1}}{\lambda_{2,2}} + \frac{k_{1}}{\lambda_{1,1}} \right)
\end{equation*}
\end{linenomath}

\noindent where $\lambda_{1,1}$ is the dominant resource allocation of the first agent to complete its work, and $\lambda_{2,1}, \; \lambda_{2,2}$ are the dominant resource allocations of the second agent to complete its work in the first and the second time intervals respectively. 

\begin{customlemma}{B.1} \label{lcp-2-agents-ef}
The LCP allocations are envy-free in expectation for the case of two agents.
\end{customlemma}

\begin{proof}
From Lemma \ref{cost-prod1} and Theorem \ref{cost-prodLCP} we know that the LCP allocation is either the Shortest Job First allocation or an optimal shared allocation. We prove that the LCP allocation $\mathbf{A}^{\textnormal{LCP}}(\cdot)$ satisfies envy-freeness in expectation for each of the two types of the LCP allocations.

If the LCP allocation is of the form of Shortest Job First (SJF) allocation for the given instance of $\mathbf{W}$, that is, $\mathbf{A}^{\textnormal{LCP}}(\cdot) = \mathbf{A}^{\textnormal{SJF}}(\cdot)$, then from Lemma \ref{sjf-ef} we know that EF in expectation is satisfied.

Now, we consider the case where the LCP allocation is a shared allocation in the first round corresponding to one of the extreme points of the Pareto frontier and a non-wasteful allocation (defined in the proof of Theorem \ref{drf-w_properties}) for the second agent in the second round, that is, $\mathbf{A}^{\textnormal{LCP}}(\cdot) = \mathbf{A}^{\textnormal{shared}}(\cdot)$.
Let us assume for the sake of contradiction that an agent 1 envies agent 2 in this shared allocation. From the definition of envy given in Lemma \ref{sjf-ef} we know that agent 1 would be able to complete its required work by switching to the shared allocation of agent 2. Thus,
\begin{linenomath}
\begin{equation*}
k_1 \leq w_{1}(\mathbf{A_{2}}(\cdot)) \leq w_{2}(\mathbf{A_{2}}(\cdot)) = k_2.    
\end{equation*}
\end{linenomath}

Let $\sigma$ be the relative rate of progress agent 1 can make with resources allocated in according to agent 2's requirements.  That is, $\sigma = \min_r d_{2r} / d_{1r}$.  Agent 1's faster rate of progress with agent 2's allocation means that $\lambda_{11} < \sigma\lambda_{21} \leq \lambda_{21}$.  Since $\lambda_{11} + \sigma\lambda_{21} \leq 1$ by the resource constraint on agent 1's dominant resource, $\lambda_{11} < 0.5$

Consider instead the allocation which gives agent 1 all resources and then agent 2 all resources.  (Since $k_1 \leq k_2$ this is a SJF allocation.)  Let $t_i$ be the completion times under the original allocation and $t_i'$ be those under the new allocation.  Then,
\begin{linenomath}
\begin{equation*}
t_1't_2' < (0.5 t_1)(k_1 + k_2) \leq t_1k_2 \leq t_1t_2.    
\end{equation*}
\end{linenomath}

This contradicts the original allocation minimizing the cost product.
\end{proof}

 Lemma \ref{lcp-1res_and_2agents} in the main manuscripts states the summary of the analysis of the LCP mechanism for a single resource, and the LCP mechanism with two agents.
 
\begin{customlemma} {6}
For single resource allocations, the LCP mechanism allocates the resource in the form of the Shortest Job First (SJF) and the allocations are envy-free in expectation. The LCP allocations are envy-free in expectation when allocating multiple resources for two agents.
\end{customlemma}

\begin{proof}
The proof of this Lemma follows directly from Lemmas \ref{sjf-lcp}, \ref{sjf-ef}, Corollary \ref{single-res-ef}, and Lemma \ref{lcp-2-agents-ef}.
\end{proof}


\subsection{Note on Experimental Simulations}

The experiments were conducted on Intel® Core™ i5-8500U CPU with four cores operating at a base frequency of 1.60 GHz in a 12 GB RAM system. All the code was programmed in Python and executed on a Linux based operating system. 

\end{document}